\documentclass[11pt,envcountsame,runningheads]{llncs}
\usepackage[a4paper,margin=2.5cm]{geometry}

\usepackage[T1]{fontenc}
\usepackage[utf8]{inputenc}
\usepackage{lmodern}
\usepackage{microtype}
\input{glyphtounicode}
\usepackage{amsmath}
\usepackage{amsfonts}
\usepackage{amssymb}

\usepackage{mathtools}
\mathtoolsset{mathic=true}

\usepackage[ruled,vlined,linesnumbered]{algorithm2e}
\SetKw{KwDownTo}{downto}

\usepackage{graphicx}
\usepackage[inline,shortlabels]{enumitem}
\setlist[enumerate]{leftmargin=*}

\usepackage[caption=false]{subfig}

\usepackage{tikz}
\pgfdeclarelayer{background}
\pgfsetlayers{background,main}

\DeclareRobustCommand{\qed}
{%
  \ifmmode \square%
  \else%
    \leavevmode\unskip\penalty9999 \hbox{}\nobreak\hfill%
    \quad\hbox{$\square$}%
  \fi%
}%

\let\oldcite\cite
\renewcommand{\cite}[1]{\textup{\oldcite{#1}}}

\newcommand{\calA}{\mathcal{A}}
\newcommand{\calB}{\mathcal{B}}
\newcommand{\calC}{\mathcal{C}}
\newcommand{\calE}{\mathcal{E}}
\newcommand{\calF}{\mathcal{F}}
\newcommand{\calI}{\mathcal{I}}
\newcommand{\calO}{\mathcal{O}}
\newcommand{\calS}{\mathcal{S}}
\newcommand{\calX}{\mathcal{X}}

\newcommand{\bbE}{\mathbb{E}}
\newcommand{\bbS}{\mathbb{S}}
\newcommand{\bbV}{\mathbb{V}}
\newcommand{\bbY}{\mathbb{Y}}
\newcommand{\bbZ}{\mathbb{Z}}

\newcommand{\frkX}{\mathfrak{X}}
\newcommand{\frkY}{\mathfrak{Y}}
\newcommand{\frkZ}{\mathfrak{Z}}

\makeatletter
\newcommand{\ie}{i.\,e.\@ifnextchar{,}{}{~}}
\makeatother

\let\spGamma\Gamma
\renewcommand{\Gamma}{\mathrm{\spGamma}}

\newcommand{\uparr}{^{\raisebox{0.1ex}{$\scriptscriptstyle \uparrow$}}}

\title
{%
    Computing the Union Join and Subset Graph of Acyclic Hypergraphs in Subquadratic Time%
}

\titlerunning
{%
    Union Join and Subset Graph for Acyclic Hypergraphs%
}

\author
{%
    Arne Leitert
}

\institute{}

\DeclareMathOperator{\2Sec}{2Sec}

\begin{document}

\maketitle

\begin{abstract}
We investigate the two problems of computing the union join graph as well as computing the subset graph for acyclic hypergraphs and their subclasses.
In the \emph{union join graph}~$G$ of an acyclic hypergraph~$H$, each vertex of~$G$ represents a hyperedge of~$H$ and two vertices of~$G$ are adjacent if there exits a join tree~$T$ for~$H$ such that the corresponding hyperedges are adjacent in~$T$.
The \emph{subset graph} of a hypergraph~$H$ is a directed graph where each vertex represents a hyperedge of~$H$ and there is a directed edge from a vertex~$u$ to a vertex~$v$ if the hyperedge corresponding to~$u$ is a subset of the hyperedge corresponding to~$v$.

For a given hypergraph~$H = (V, \calE)$, let $n = |V|$, $m = |\calE|$, and $N = \sum_{E \in \calE} |E|$.
We show that, if the Strong Exponential Time Hypothesis is true, both problems cannot be solved in $\calO \bigl( N^{2 - \varepsilon} \bigr)$ time for $\alpha$-acyclic hypergraphs and any constant $\varepsilon > 0$, even if the created graph is sparse.
Additionally, we present algorithms that solve both problems in $\calO \bigl( N^2 / \log N + |G| \bigr)$ time for $\alpha$-acyclic hypergraphs, in $\calO \bigl( N \log (n + m) + |G| \bigr)$ time for $\beta$-acyclic hypergaphs, and in $\calO \bigl( N + |G| \bigr)$ time for $\gamma$-acyclic hypergraphs as well as for interval hypergraphs, where $|G|$ is the size of the computed graph.
\end{abstract}

\section{Introduction}

A \emph{hypergraph~\( H = (V, \calE) \)} is a generalisation of a graph in which each edge~$E \in \calE$, called \emph{hyperedge}, can contain an arbitrary positive number of vertices from~$V$.
One may also see a hypergraph~$H$ as a family~$\calE$ of subsets of some set~$V$.
Indeed, we say that the family~$\calF$ of sets \emph{forms} the hypergraph~$H = (V, \calE)$ if $V = \bigcup_{S \in \calF} S$ and $\calE = \calF$.
We use $n = |V|$, $m = |\calE|$, and $N = \sum_{E \in \calE} |E|$ to respectively denote the cardinality of the vertex set, the cardinality of the hyperedge set, and the total size of all hyperedges of~$H$.

\subsection{Acyclic Hypergraphs}

A tree~$T$ is called a \emph{join tree} for~$H$ if the hyperedges of~$H$ are the nodes of~$T$ and, for each vertex~$v \in V$, the hyperedges containing~$v$ induce a subtree of~$T$.
That is, if $v \in E_i \cap E_j$, then $v$ is contained in each hyperedge (\ie, node of~$T$) on the path from $E_i$ to~$E_j$ in~$T$.
A hypergraph is \emph{acyclic} if it admits a join tree.
There is a linear-time algorithm which determines if a given hypergraph is acyclic and, in that case, constructs a corresponding join tree for it~\cite{TarjanYannak1984}.

Acyclic hypergraphs have various applications.
They are, for example, a desired structure when designing relational databases~\cite{BeeFagMaiYan1983}.
There is also a close relation between acyclic hypergraphs and chordal as well as dually chordal graphs.
Namely, a graph is chordal if and only if its maximal cliques form an acyclic hypergraph~\cite{Gavril1974}, and a graph is dually chordal if and only if its closed neighbourhoods form an acyclic hypergraph~\cite{BraDraCheVol1998}.

Tree-decompositions are another application.
The idea is to decompose a graph~$G = (V, E)$ into multiple induced subgraphs, usually called \emph{bags}, where each vertex can be in multiple bags.
The set of bags~$\calB$ forms a tree~$T$ in such a way that the following requirements are fulfilled:
Each vertex is in at least one bag, each edge is in at least one bag, and $T$ is a join tree for the hypergraph~$(V, \calB)$.
Usually tree-decompositions are considered with additional restrictions.
The most known is called \emph{tree-width}; it limits the maximum cardinality of each bag.
For a graph class with bounded tree-width, many NP-complete problems can be solved in polynomial or even linear time.
Alternatively, one may limit the distances between vertices inside a bag.
Such a tree-decomposition can be used, for example, for constructing tree-spanners~\cite{DouDraGavYan2007,DraganKohler2014} and efficient routing schemes~\cite{Dourisboure2005}.

An inclusion-maximal subset of vertices of a graph~$G$ is called an \emph{atom} if it induces a connected subgraph of~$G$ without a clique separator.
It is known that the atoms of a graph form an acyclic hypergraph~\cite{Leimer1993}.
The corresponding join tree is then called \emph{atom tree}.

The most general acyclic hypergraphs are called \emph{\( \alpha \)\nobreakdash-acyclic} (\ie., each acyclic hypergraph is $\alpha$-acyclic).
They are closely related to chordal graphs and to dually chordal graphs.
Subclasses of $\alpha$-acyclic hypergraphs are \emph{\( \beta \)-acyclic} hypergraphs which are closely related to strongly chordal graphs and \emph{\( \gamma \)\nobreakdash-acyclic} hypergraphs which are closely related to ptolemaic graphs (graphs that are chordal and distance-hereditary).
We also consider \emph{interval} hypergraphs.
These are acyclic hypergrpahs for which one of their join trees forms a path.
As the name suggests, they are closely related to interval graphs.
We give formal definitions and more information about each subclass later in their respective sections.

A class of hypergraphs closely related to acyclic hypergraphs are so-called \emph{hypertrees}.
These hypergraphs are defined in the same way as acyclic hypergraphs, except that the roles of vertices and hyperedges are exchanged.
That is, a hypergraph is a hypertee if its vertices admit a tree~$T$ such that each hyperedge induces a subtree of~$T$.
The hypergraph resulting from exchanging the roles of vertices and hyperedges is called the \emph{dual} hypergraph.
(See Section~\ref{sec:preliminaries} for a more formal definition.)
Subsequently, a hypergraph is a hypertree if and only if it is the dual of an acyclic hypergraph.

Figure~\ref{fig:hierarchy} shows the hierarchy of acyclic hypergraphs.
See \textsc{Brandstädt} and \textsc{Dragan}~\cite{BrandsDragan2015} for a summary of known properties of acyclic hypergraphs as well as their relations to various graph classes.

\begin{figure}
    \centering
    \begin{tikzpicture}
[
    every node/.style =
    {
        anchor = base,
        fill = white,
        inner sep = 5pt,
        draw,
    },
    yscale = 1.33,
    xscale = 1.25,
]

\node (aacy) at (-1,  0) {$\alpha$-acyclic};
\node (hypt) at ( 1,  0) {hypertree};
\node (acht) at ( 0, -1) {$\alpha$-acyclic $\cap$ hypertree};
\node (bacy) at ( 0, -2) {$\beta$-acyclic};
\node (gacy) at (-1, -3) {$\gamma$-acyclic};
\node (intv) at ( 1, -3) {interval};

\begin{pgfonlayer}{background}
\begin{scope}
[
    -latex,
    shorten > = 1pt
]

\draw (gacy.center) -- (bacy);
\draw (intv.center) -- (bacy);
\draw (bacy.center) -- (acht.south);
\draw (acht.center) -- (aacy.south);
\draw (acht.center) -- (hypt.south);

\end{scope}
\end{pgfonlayer}
\end{tikzpicture}
    \caption
    {%
        Hierarchy of acyclic hypergraphs.
        An edge from class~$X$ to class~$Y$ states that $X$ is a proper subset of~$Y$.
    }
    \label{fig:hierarchy}
\end{figure}
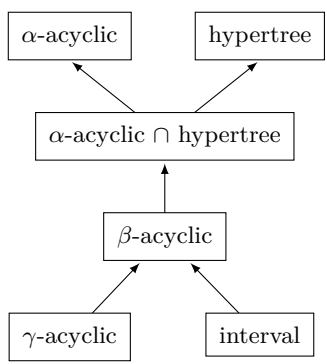

\subsection{Union Join Graph}

Note that the join tree of an acyclic hypergraph is not always unique.
For example, each tree with $n$~nodes is a valid join tree for the hypergraph formed by~$\bigl \{ \{ 0, 1 \}, \{ 0, 2 \}, \ldots, \{ 0, n \}  \bigr \}$.
The \emph{union join graph}~$G$ of a given acyclic hypergraph~$H$ is the union of all its join trees.
That is, each vertex of~$G$ represents a hyperedge of~$H$ and two vertices of~$G$ are adjacent if there exits a join tree~$T$ for~$H$ such that the corresponding hyperedges are adjacent in~$T$.
The union join graph of a hypergraph~$H$ may also be called \emph{clique graph} if $H$ represents the maximal cliques of a chordal graph~\cite{GaliHabiPaul1995,HabibStacho2012}, or \emph{atom graph} if $H$ represents the atoms of some graph~\cite{KabPinLelBer2007}.
In~\cite{BerrySimone2016}, \textsc{Berry} and \textsc{Simonet} present algorithms which compute the union join graph of an acyclic hypergraph in $\calO(Nm)$ time.

\subsection{Subset Graph}

The \emph{subset graph} of a hypergraph~$H$ is a directed graph~$G$ where each vertex represents a hyperedge of~$H$ and there is a directed edge from a vertex~$u$ to a vertex~$v$ if the hyperedge corresponding to~$u$ is a subset of the hyperedge corresponding to~$v$.
\textsc{Pritchard} presents an algorithm in~\cite{Pritchard1999} that computes the subset graph for a given hypergraph in $\calO \bigl( N^2 / \log N \bigr)$ time.
They also show that any subset graph has at most $\calO \bigl( N^2 / \log^2 N  \bigr)$ many edges.
There are various publications that present algorithms for special cases and different computational models; see for example \cite{Elmasry2009,Pritchard1991} and the work cited therein.

The \emph{Strong Exponential Time Hypothesis}, \emph{SETH} for short, states that there is no algorithm that solves the Boolean satisfiability problem (without limitation on clause size) for some constant $\varepsilon > 0$ in $\calO \bigl( (2 - \varepsilon)^n \bigr)$ time where $n$ is the number of variables in the given instance.
A function~$f(n)$ is called \emph{truly subquadratic} if $f(n) \in \calO \bigl( n^{2 - \varepsilon} \bigr)$ for some constant~$\varepsilon > 0$.
\textsc{Borassi}~et~al.\,\cite{BoraCresHabi2016} show that, if SETH holds, then there is no algorithm to compute the subset graph of an arbitrary hypergraph in truly subquadratic time, even if the output is sparse.
Note that the results in \cite{BoraCresHabi2016} and~\cite{Pritchard1999} are not conflicting, since $N^{2 - \varepsilon} \in o \bigl( N^2 / \log N \bigr)$.

\subsection{Our Contribution}

In this paper, we investigate the two problems of computing the union join graph as well as computing the subset graph for acyclic hypergraphs and their subclasses.
We show in Section~\ref{sec:alphaAcyclic} that there is a close relation between both problems by presenting reductions in both directions.
It then follows that the result by \textsc{Borassi}~et~al.\ still holds when restricted to $\alpha$-acyclic hypergraphs and also applies to computing a union join graph.
We then develop efficient algorithms to solve both problems for acyclic hypergraphs and their subclasses.
In particular, we show that, if $|G|$ denotes the size of the computed graph~$G$, then both problems can be solved in $\calO \bigl( N^2 / \log N + |G| \bigr)$ time for $\alpha$-acyclic hypergraphs (Section~\ref{sec:alphaAcyclic}), in $\calO \bigl( N \log (n + m) + |G| \bigr)$ time for $\beta$-acyclic hypergaphs (Section~\ref{sec:betaAcyclic}), and in $\calO \bigl( N + |G| \bigr)$ time for $\gamma$-acyclic hypergraphs (Section~\ref{sec:gammaAcyclic}) as well as for interval hypergraphs (Section~\ref{sec:interval}).

\section{Preliminaries}
\label{sec:preliminaries}

Two graphs $G = (V, E)$ and $G' = (V', E')$ are \emph{isomorphic} if there is a bijective function~$f \colon V \rightarrow V'$ such that $uv \in E$ if and only if $f(u)f(v) \in E'$.
For simplicity, we write $G = G'$ if they are isomorphic.

Let $H = (V, \calE)$ be a hypergraph.
The \emph{incidence graph~\( \calI(H) = \bigl( U_V \cup U_\calE, E_\calI \bigr) \)} of~$H$ is a bipartite graph were $U_V$ represents the vertices of~$H$, $U_\calE$ represents the hyperedges of~$H$, and there is an edge between two vertices $u_v \in U_V$ and $u_E \in U_\calE$ if the corresponding vertex~$v$ (of~$H$) is in the corresponding hyperedge~$E$.
That is, $U_V = \{ \, u_v \mid v \in V \, \}$, $U_\calE = \{ \, u_E \mid E \in \calE \, \}$, and $E_\calI = \{ \, u_vu_E \mid v \in E \, \}$.
Note that $\bigl| E_\calI \bigr| = N$.
If not stated or constructed otherwise, the incidence graphs of all hypergraphs occurring in this paper are connected, finite, undirected, and without multiple edges.
Additionally, whenever a hypergraph is given, it is given as its incidence graph; hence, the input size is in~$\Theta(N)$.
We say two hyperedges of~$H$ are \emph{distinct} if they are represented by two different vertices in~$\calI(H)$, even if both hyperedges contain the same vertices.

Let $\bigl( U_V \cup U_\calE, E_\calI \bigr)$ be the incidence graph of some hypergraph~$H = (V, \calE)$.
One can then exchange the roles of $U_V$ and $U_\calE$ to interpret $U_\calE$ as vertices and $U_V$ as hyperedges.
We call the resulting hypergraph the \emph{dual} hypergraph of~$H$ and denote it as~$H^*$.
Observe that, by definition, $(H^*)^* = H$.

The \emph{2-section graph~\( \2Sec(H) \)} of~$H$ is the graph with the vertex set~$V$ where two vertices $u$ and~$v$ are adjacent if there is a hyperedge~$E \in \calE$ with $u, v \in E$.
The \emph{line graph~\( L(H) \)} of~$H$ is the intersection graph of its hyperedges.
That is, $L(H) = (\calE, \calE_L)$ with $\calE_L = \{ \, E_iE_j \mid E_i, E_j \in \calE; E_i \cap E_j \neq \emptyset \, \}$.
It directly follows from these definitions that $\2Sec(H) = L(H^*)$.

A sequence $\langle v_1, v_2, \ldots, v_k \rangle$ of vertices of~$H$ \emph{forms a path} in~$H$ if, for each~$i$ with $1 \leq i < k$, $H$ contains a hyperedge~$E$ with~$v_i, v_{i + 1} \in E$.
Let $X$, $Y$, and~$Z$ be sets of vertices of~$H$.
$X$ \emph{separates} $Y$ form~$Z$ if $X \neq \emptyset$ and each sequence of vertices that forms a path from $Y$ to~$Z$ in~$H$ contains a vertex from~$X$.

Let $T$ be the join tree of some acyclic hypergraph~$H$ and let $E_i$ and~$E_j$ be two hyperedges of~$H$ which are adjacent in~$T$.
We then call the set~$S = E_i \cap E_j$ a \emph{separator} of~$H$ with respect to~$T$.
If $T$ is rooted and $E_i$ is the parent of~$E_j$, we call $S\uparr(E_j) := E_i \cap E_j$ the \emph{up-separator} of~$E_j$.
Note that each separator corresponds to an edge of~$T$ and vice versa.
We call the hypergraph formed by the set of all separators of~$H$ the \emph{separator hypergraph~\( \calS(H) \)} for~$H$ with respect to~$T$.
It follows from properties~\ref{item:joinTreeEdge} and~\ref{item:intersectionOnPath} of Lemma~\ref{lem:unionJoinProperties} (see Section~\ref{sec:alphaAcyclic}) that $\calS(H)$ is always the same for a given~$H$, independent of the used join tree.

\section{$\alpha$-Acyclic Hypergraphs}
\label{sec:alphaAcyclic}

In this section, we investigate the problems of computing a union join graph and computing a subset graph for the most general case of acyclic hypergraphs.
We first show that computing these graphs cannot be done in truly subqadratic time if the SETH is true.
For that, we use a problem called \emph{Sperner Family} problem.
It asks whether a family of sets contains two sets $S$ and~$S'$ such that $S \subseteq S'$.
If the SETH is true, then there is no algorithm that solves it truly subquadratic time~\cite{BoraCresHabi2016}.
Afterwards, we give an algorithm that allows to quickly compute the union join graph if a fast algorithm for the subset graph problem is given.
Lastly, we give some additional notes on the Sperner Family problem and its generalisation.

\subsection{Hardness Results}

Let $\calF = \{ S_1, S_2, \ldots, S_m \}$ be a family of sets.
We create an acyclic hypergraph~$H$ from~$\calF$ as follows.
Create a new vertex~$u$ (\ie, $u$ is not contained in any set~$S_i$) and, for each set~$S_i$, create a hyperedge~$E_i = S_i \cup \{ u \}$.
Additionally, create a hyperedge~$\calS$ which is the union of all hyperedges~$E_i$.
Formally, we have that $H = (V, \calE)$ with $V = \calS$ and $\calE = \big \{ \, E_i \bigm| S_i \in \calF \, \big \} \cup \{ \calS \}$.
One can create a join tree~$T$ for~$H$ by starting with $\calS$ and then making each hyperedge~$E_i$ adjacent to it.
Thus, $H$ is acyclic.
Note that one can create $H$ and~$T$ from~$\calF$ in linear time.

For the remainder of this subsection, assume that we are given a family~$\calF$, a hypergraph~$H$, and a corresponding join tree~$T$ for~$H$ as defined above.
Our results in this subsection are based on the following observation.

\begin{lemma}
\label{lem:subsetIffTreeEdge}
\( \calF \) contains two distinct sets \( S_i \) and~\( S_j \) with \( S_i \subseteq S_j \) if and only if there is a join tree for \( H \) that contains the edge \( E_iE_j \).
\end{lemma}

\begin{proof}
First, assume that $\calF$ contains two distinct sets $S_i$ and~$S_j$ with $S_i \subseteq S_j$.
In that case, we can create a new join tree~$T'$ as follows.
Remove the edge~$E_i\calS$ from~$T$ and make $E_i$ adjacent to~$E_j$ instead.
Since $S_i \subseteq S_j$, each element~$x \in E_i \cap \calS$ is also contained in~$E_j$.
Thus, $T'$ is a join tree for $H$ and contains the edge~$E_iE_j$.

Next, assume that there is a join tree~$T'$ for $H$ with the edge~$E_iE_j$.
Without loss of generality, let $E_j$ be closer to~$\calS$ in~$T'$ than $E_i$.
Recall that $E_i \subseteq \calS$.
Therefore, by properties of join trees, each vertex in~$E_i$ is also in~$E_j$.
It then directly follows from the construction of~$H$ that $S_i \subseteq S_j$.
\qed
\end{proof}

We use the Sperner Family problem to show that there is no truly subquadratic-time algorithm to compute the union join graph of a given acyclic hypergraph.
To do so, we first show the following.

\begin{lemma}
\label{lem:hardnessUniqueJoinTree}
If the SETH is true, then there is no algorithm which decides in \( \calO \big( N^{2 - \varepsilon} \big) \) time whether or not a given acyclic hypergraph has a unique join tree.
\end{lemma}

\begin{proof}
Recall that we can create a join tree~$T$ for~$H$ by making each hyperedge~$E_i$ adjacent to the hyperedge~$\calS$.
To prove Lemma~\ref{lem:hardnessUniqueJoinTree}, we show that $\calF$ contains two distinct sets $S_i$ and~$S_j$ with $S_i \subseteq S_j$ if and only if $T$ is not a unique join tree for~$H$.

First, assume that $\calF$ contains two such sets $S_i$ and~$S_j$.
In that case, Lemma~\ref{lem:subsetIffTreeEdge} implies that there is a join tree~$T'$ for $H$ with the edge~$E_iE_j$.
Since $E_iE_j$ is not an edge in~$T$, $T$ is not unique.
Next, assume that $T$ is not unique.
Then, there is a join tree~$T'$ and a hyperedge~$E_i$ such that $E_i$ is not adjacent to~$\calS$ in~$T'$.
Hence, $E_i$ is adjacent to some hyperedge~$E_j$ that is closer to~$\calS$ in~$T'$ than~$E_i$.
Since $E_i \subseteq \calS$, properties of join trees imply that $E_i \subseteq E_j$.
Subsequently, due to Lemma~\ref{lem:subsetIffTreeEdge}, $S_i \subseteq S_j$.

It follows that a truly subquadratic-time algorithm which determines if an acyclic hypergraph has a unique join tree would imply an equally fast algorithm to solve the Sperner Family problem for any family of sets.
\qed
\end{proof}

Note that, by definition of a union join graph, $H$ has a unique join tree if and only if the union join graph of~$H$ is a tree.
Therefore, we get the following.

\begin{theorem}
\label{theo:hardnessUnionJoinGraph}
If the SETH is true, then there is no algorithm which constructs the union join graph of a given acyclic hypergraph in \( \calO \big( N^{2 - \varepsilon} \big) \) time, even if that graph is sparse.
\end{theorem}

We now show that computing the subset graph of an acyclic hypergraph is as hard as computing the subset graph for a general family of sets.

\begin{theorem}
\label{theo:hardnessSubsetGraph}
If the SETH is true, then there is no algorithm which constructs the subset graph of a given acyclic hypergraph in truly subquadratic time.
\end{theorem}

\begin{proof}
Let $G$ be the subset graph for~$H$ and $G_\calF$ be the subset graph for~$\calF$.
Since, by construction of~$H$, $E_i \subseteq E_j$ if and only if $S_i \subseteq S_j$, $G$ contains the edge~$(E_i, E_j)$ if and only if $G_\calF$ contains the edge~$(S_i, S_j)$.
We can therefore construct $G_\calF$ from~$G$ by simply removing the vertex representing~$\calS$ from~$G$ (and its incident edges).

Recall that we can construct~$H$ from~$\calF$ in linear time.
Therefore, a truly subquadratic-time algorithm to construct the subset graph of a given acyclic hypergraph would imply an equally fast algorithm to construct a subset graph of a given family of sets.
\qed
\end{proof}

\subsubsection{Note on Hypertrees.}

Observe that, in the hypergraph~$H$ as constructed above, each hyperedge contains the vertex~$u$.
We can therefore create a tree~$T$ by making each other vertex a leaf adjacent to~$u$.
Each hyperedge of~$H$ now induces a subtree of~$T$, \ie, $H$ is a hypertree.

It follows that Lemma~\ref{lem:hardnessUniqueJoinTree} and Theorem~\ref{theo:hardnessUnionJoinGraph} still hold if the given hypergraph is both acyclic and a hypertree.
Therefore, there is no truly subquadratic-time algorithm which, in general, computes the union join graph of such a hypergraph or determines if has a unique join tree.

\subsection{Union Join Graph via Subset Graph}
\label{sec:alphaSubsetG}

In the previous subsection, we show how to compute the subset graph using the union join graph of an acyclic hypergraph.
We now present an algorithm that computes the union join graph of a given acyclic hypergraph with the help of a subset graph.
The runtime of our algorithm then depends on the runtime required to compute that subset graph.

For the remainder of this subsection, assume that we are given an acyclic hypergraph~$H = (V, \calE)$ and let $G$ be the union join graph of~$H$ (with for us unknown edges).
Lemma~\ref{lem:unionJoinProperties} below gives various characterisations for~$G$.

\begin{lemma}
\label{lem:unionJoinProperties}
For any distinct \( E_i, E_j \in \calE \), the following are equivalent.
\begin{enumerate}[(i)]
    \item
    \label{item:unionJoinGrpahEdge}
        \( E_iE_j \) is an edge of~\( G \).
    \item
    \label{item:joinTreeEdge}
        \( H \) has a join tree with the edge~\( E_iE_j \).
    \item
    \label{item:intersectionOnPath}
        Each join tree~\( T \) of~\( H \) has an edge~\( E_i'E_j' \) on the path from \( E_i \) to~\( E_j \) in~\( T \) such that \( E_i \cap E_j = E_i' \cap E_j' \).
    \item
    \label{item:edgesOnBothSideOfSep}
        Each join tree~\( T \) of~\( H \) has a separator~\( S \) on the path~\( P_{ij} \) from \( E_i \) to~\( E_j \) in~\( T \) with \( S \subseteq S_i \) and \( S \subseteq S_j \) where \( S_i \) and~\( S_i \) are the separators in~\( P_{ij} \) which are respectively closest to \( E_i \) and~\( E_j \).
    \item
    \label{item:edgeSeparates}
        \( E_i \cap E_j \) separates \( E_i \setminus E_j \) from~\( E_j \setminus E_i \).
\end{enumerate}
\end{lemma}

Most of the properties in Lemma~\ref{lem:unionJoinProperties} repeat, generalise, or paraphrase existing results (see~\cite{BerrySimone2016,GaliHabiPaul1995,HabibStacho2012}).
Property~\ref{item:edgesOnBothSideOfSep} is, to the best of our knowledge, a new observation.
For completeness, however, we prove all of them.

\begin{proof}
By definition of~$G$, properties \ref{item:unionJoinGrpahEdge} and~\ref{item:joinTreeEdge} are equivalent.
It follows from properties of join trees that \ref{item:joinTreeEdge} implies~\ref{item:edgeSeparates}.

We next show that \ref{item:edgeSeparates} implies~\ref{item:intersectionOnPath}.
Assume that $E_i$ and~$E_j$ are not adjacent in a join tree~$T$.
Then there is a path $\langle E_i = X_1, X_2, \ldots, X_k = E_j \rangle$ of hyperedges from $E_i$ to~$E_j$ in~$T$.
For each~$p$ with $1 \leq p < k$, let $S_p = X_p \cap X_{p + 1}$ be the separator corresponding to the edge~$X_pX_{p + 1}$ of~$T$.
By properties of join trees, $E_i \cap E_j \subseteq S_p$ for each~$S_p$.
Now assume that each~$S_p$ contains a vertex~$v_p \notin E_i \cap E_j$.
Then, $\langle v_1, v_2, \ldots, v_{k - 1} \rangle$ would form a path in~$H$ from~$v_1 \in E_i \setminus E_j$ to~$v_{k - 1} \in E_j \setminus E_i$.
That contradicts with property~\ref{item:edgeSeparates}.
Therefore, there is at least one separator~$S_p$ with $S_p \subseteq E_i \cap E_j$, \ie, there is an edge~$X_pX_{p + 1}$ in~$T$ with~$E_i \cap E_j = X_p \cap X_{p + 1}$.

To show that \ref{item:intersectionOnPath} implies~\ref{item:joinTreeEdge}, consider a join tree~$T$ where $E_i$ and~$E_j$ are not adjacent.
We can create a join tree~$T'$ by removing the edge~$E_i'E_j'$ and adding the edge~$E_iE_j$ instead.
Since $E_i$ and~$E_j$ are on different sides of~$E_i'E_j'$ in~$T$, $T'$ is also a tree.
Additionally, because $E_i \cap E_j = E_i' \cap E_j'$, $T'$ is a valid join tree for~$H$.

It remains to show that \ref{item:edgesOnBothSideOfSep} is equivalent to~\ref{item:intersectionOnPath}.
We first assume property~\ref{item:intersectionOnPath}.
Let $S = E_i \cap E_j$ be a separator on the path from $E_i$ to~$E_j$ in some join tree~$T$.
Since, by properties of join trees, each vertex in~$S = E_i \cap E_j$ is also in $S_i$ and~$S_j$, it follows that $S \subseteq S_i$ and~$S \subseteq S_j$.
Now assume property~\ref{item:edgesOnBothSideOfSep}.
Because $S \subseteq S_i \subseteq E_i$ and~$S \subseteq S_j \subseteq E_j$, it is also the case that $S \subseteq E_i \cap E_j$.
Since $S$ is on the path from $E_i$ to~$E_j$ in~$T$, each vertex that is in both $E_i$ and~$E_j$ also has to be in~$S$, \ie, $S \supseteq E_i \cap E_j$.
Therefore, $S = E_i \cap E_j$.
\qed
\end{proof}

Based on Lemma~\ref{lem:unionJoinProperties}, we can construct $G$ as follows.
Compute a join tree~$T$ for~$H$, the separator hypergraph~$\calS(H)$ (with respect to~$T$), and its subset graph~$G_\calS$.
Next, use $G_\calS$ to find all triples $S_i, S_j, S$ of separators which satisfy property~\ref{item:edgesOnBothSideOfSep} of Lemma~\ref{lem:unionJoinProperties}.
Since their corresponding hyperedges are then adjacent in some join tree of~$H$, make the corresponding vertices adjacent in~$G$.

Before analysing our approach further, we address some needed preprocessing.
Assume that $H$ contains two hyperedges $E_i$ and~$E_j$ which are not adjacent in~$T$, but are adjacent in some other join tree.
There might then be multiple separators~$S$ on the path from $E_i$ to~$E_j$ in~$T$ which satisfy property~\ref{item:edgesOnBothSideOfSep} of Lemma~\ref{lem:unionJoinProperties}.
Our algorithm would, therefore, add the edge~$E_iE_j$ to~$G$ multiple times, once for each such~$S$.
While it is easy to remove redundant edges from~$G$ afterwards, we still want to ensure that the time needed to create and remove these edges does not become too much.
To achieve that, Algorithm~\ref{algo:modifyT} modifies $T$ such that each hyperedge becomes adjacent to its highest possible ancestor in~$T$.
As by-product, Algorithm~\ref{algo:modifyT} also computes the up-separator of each hyperedge (and, thus, the separator hypergraph~$\calS(H)$).

\begin{algorithm}
\caption
{%
    Modifies the join tree of a given acyclic hypergraph such that each hyperedge becomes adjacent to its highest possible ancestor.
}
\label{algo:modifyT}

\KwIn
{%
    An acyclic hypergraph~$H = (V, \calE)$ and a join tree~$T$ for~$H$.
}

\KwOut
{%
    A modified join tree~$T'$ for~$H$ and the separator hypergraph~$\calS(H)$.
}

Root $T$ in an arbitrary hyperedge~$R$ and then run a pre-order on~$T$.
Let $\sigma = \langle R = E_1, E_2, \ldots, E_m \rangle$ be the resulting order.
\label{line:modT_preOrder}

For each vertex~$v$, set $\lambda(v) := \min \{ \, i \mid v \in E_i \, \}$.
\label{line:modT_vertexLbl}

\For
{%
    \( i := 2 \) \KwTo \( m \)
}
{%
    Set $S\uparr(E_i) := \bigl \{ \, v \in E_i \bigm| \lambda(v) < i \, \bigr \}$.
    \label{line:modT_upSep}

    Let $j = \max \bigl \{ \, \lambda(v) \bigm| v \in S\uparr(E_i) \, \bigr \}$ and make $E_j$ the parent of~$E_i$.
    \label{line:modT_newParent}
}

Let $\calS(H)$ be the hypergraph formed by the family~$\bigl \{ \, S\uparr(E_i) \bigm| E_i \in \calE, E_i \neq R \, \bigr \}$.
\end{algorithm}

\begin{lemma}
\label{lem:modifyTRuntime}
Algorithm~\ref{algo:modifyT} runs in linear time.
\end{lemma}

\begin{proof}
Line~\ref{line:modT_preOrder} runs in $\calO(m)$ time, since the nodes of~$T$ are the hyperedges of~$H$.
Recall that $H$ is given as an incidence graph~$\calI(H)$.
Hence, the following are equivalent (with respect to runtime):
(i)~for each vertex, iterating over all hyperedges containing it;
(ii)~for each hyperedge, iterating over all vertices it contains; and
(iii)~iterating over all edges of~$\calI(H)$.
Therefore, line~\ref{line:modT_vertexLbl}, line~\ref{line:modT_upSep}, and line~\ref{line:modT_newParent} (and subsequently Algorithm~\ref{algo:modifyT}) run in $\calO(N)$ total time.
\qed
\end{proof}

\begin{lemma}
The tree~\( T' \) created by Algorithm~\ref{algo:modifyT} is a valid join tree for~\( H \).
\end{lemma}

\begin{proof}
Let $T_i$ be the tree after processing~$E_i$, \ie, $T = T_1$ and $T_m = T'$.
Thus, $T_1$ is a valid join tree for~$H$.
Assume, by induction, that $T_{i - 1}$ (with $i \geq 2$) is a valid join tree for~$H$ too.
Recall that, by definition of join trees, the set of hyperedges containing a vertex~$v$ form a subtree~$T_v$ of~$T$.
The roots of all such $T_v$ where $v \in S\uparr(E_i)$ are ancestors of~$E_i$ in~$T$ and, thus, form a path.
By definition of~$j$ (line~\ref{line:modT_newParent}), $E_j$ is the lowest of such roots in~$T$.
It therefore follows that $S\uparr(E_i) \subseteq E_j$.
Subsequently, for each $v \in S\uparr(E_i)$, the hyperedges containing $v$ still form a subtree of~$T_i$ after changing the parent of~$E_i$ if they did so in~$T_{i - 1}$.
Note that each subtree~$T_u$ of a vertex~$u \notin S\uparr(E_i)$ remains unchanged, since it does not contain the edge~$E_iE_k$.
Therefore, for each vertex, the hyperedges containing it form a subtree of~$T_i$ and, thus, $T_i$ is a join tree for~$H$.
\qed
\end{proof}

\begin{lemma}
\label{lem:atMost2Sep}
Let \( E_i \) and~\( E_j \) be two hyperedges of~\( H \), \( T' \) be the tree computed by Algorithm~\ref{algo:modifyT}, and \( P_{ij} \) be the path from \( E_i \) to~\( E_j \) in~\( T' \).
Additionally, let \( S_i \) and~\( S_j \) be the separators on~\( P_{ij} \) which are closest to \( E_i \) and~\( E_j \), respectively.
There are at most two separators~\( S \) on~\( P_{ij} \) such that \( S \subseteq S_i \) and \( S \subseteq S_j \).
\end{lemma}

\begin{proof}
Let $E_k$ be the lowest common ancestor of $E_i$ and~$E_j$ in~$T'$.
Although $T'$ has a potentially different structure than~$T$, it is still the case that the parent of a hyperedge in~$T'$ was an ancestor of it in~$T$.
Thus, $k \leq i, j$.
Note that $P_{ij}$ goes through~$E_k$ and let $P_{ik}$ and~$P_{kj}$ be the respective subpaths of~$P_{ij}$.
If $P_{ij}$ contains more than two separators~$S$ as defined in Lemma~\ref{lem:atMost2Sep}, at least two of them are either part of $P_{ik}$ or~$P_{kj}$.
Without loss of generality, let them be on~$P_{kj}$ and let $S$ be the lowest such separator.
Additionally, let $X$ be the hyperedge directly below~$S$, \ie, $S\uparr(X) = S$.
It follows that $X$ is not adjacent to~$E_k$ in~$T'$.

Since $S \subseteq S_i$, each vertex in~$S$ is in all hyperedges on the path from $X$ to~$E_i$ in~$T'$, including~$E_k$.
Therefore, $S \subseteq E_k$ and $\max \bigl \{ \, \lambda(v) \bigm| v \in S \, \bigr \} \leq k$.
That is a contradiction, since Algorithm~\ref{algo:modifyT} would have made $X$ adjacent to~$E_k$ or one of its ancestors.
\qed
\end{proof}

Algorithm~\ref{algo:alphaUnionJoin} now implements the approach described above.
It also uses Algorithm~\ref{algo:modifyT} as preprocessing.
Therefore, due to Lemma~\ref{lem:atMost2Sep}, the algorithm adds each edge $E_iE_j$ at most two times into~$G$.

\begin{algorithm}[!htb]
\caption
{%
    Computes the union join graph of an acyclic hypergraph.
}
\label{algo:alphaUnionJoin}

\KwIn
{%
    An acyclic hypergraph~$H = (V, \calE)$ and an algorithm~$\calA$ that computes the subset graph for a given family of sets.
}

\KwOut
{%
    The union join graph~$G$ of~$H$.
}

Find a join tree for~$H$ (see~\cite{TarjanYannak1984}) and call Algorithm~\ref{algo:modifyT}.
Let $T$ be the resulting join tree and $\calS$ the resulting family of separators (\ie, the hyperedges of~$\calS(H)$).
\label{line:alphaUJ_preprocc}

Use algorithm~$\calA$ to compute the subset graph~$G_\calS$ of~$\calS$.
\label{line:alphaUJ_callA}

Create a new graph~$G = (\calE, E_G)$ with $E_G = \emptyset$.
\label{line:alphaUJ_createG}

\ForEach
{%
    \( S \in \calS \)%
    \label{line:alphaUJ_sepLoop}
}
{
    Use $G_\calS$ to determine all separators~$S'$ with $S \subseteq S'$ (including $S$ itself).
    \label{line:alphaUJ_detSubsets}

    For each such $S'$, let $EE'$ be the edge of~$T$ which $S'$ represents and let $E$ be the hyperedge farther away from~$S$ in~$T$.
    Add $E$ to a set~$\bbE$ of hyperedges.
    If $S$ and~$S'$ represent the same edge of~$T$, also add~$E'$.
    \label{line:alphaUJ_detEdges}

    Partition $\bbE$ into two sets $\bbE_1$ and~$\bbE_2$ based on which side of~$S$ they are in~$T$.
    \label{line:alphaUJ_partition}

    For each pair $E_1, E_2$ with $E_1 \in \bbE_1$ and $E_2 \in \bbE_2$, add $E_1E_2$ into~$E_G$.
    \label{line:alphaUJ_addEdges}
}
\end{algorithm}

\begin{theorem}
Algorithm~\ref{algo:alphaUnionJoin} computes the union join graph~\( G \) of a given acyclic hypergraph~\( H \) in \( \calO \bigl( T_\calA(H) + N + |G| \bigr) \) time where \( T_\calA(H) \) is the runtime of a given algorithm~\( \calA \) with the separator hypergraph of~\( H \) as input.
\end{theorem}

\begin{proof}[Correctness]
Let $E_i$ and~$E_j$ be two hyperedges of~$H$.
Additionally, let $S_i$ and~$S_j$ be the separators on the path from $E_i$ to~$E_j$ in~$T$ (computed in line~\ref{line:alphaUJ_preprocc}) which are closest to $E_i$ and~$E_j$, respectively.
We show the correctness of Algorithm~\ref{algo:alphaUnionJoin} by showing that $E_iE_j$ is an edge of~$G$ if and only if there is a join tree for~$H$ with the edge~$E_iE_j$.

First, assume that there is a join tree for~$H$ with the edge~$E_iE_j$.
Lemma~\ref{lem:unionJoinProperties} then implies that there is a separator~$S \in \calS$ such that $S \subseteq S_i$, $S \subseteq S_j$, and $E_i$ and~$E_j$ are on different sides of~$S$ in~$T$.
Therefore, when processing $S$, the algorithm finds $S_i$ and~$S_j$ (line~\ref{line:alphaUJ_detSubsets}) and consequently adds $E_i$ and~$E_j$ into~$\bbE$ (line~\ref{line:alphaUJ_detEdges}).
Since both hyperedges are on different sides of~$S$, Algorithm~\ref{algo:alphaUnionJoin} then also adds the edge~$E_iE_j$ to~$G$ (line~\ref{line:alphaUJ_addEdges}).

We now assume that $E_iE_j$ is an edge of~$G$.
Note that Algorithm~\ref{algo:alphaUnionJoin} only adds edges to~$G$ in line~\ref{line:alphaUJ_addEdges}.
Thus, there is a separator~$S \in \calS$ for which the algorithm adds $E_iE_j$ to~$G$.
For that~$S$, one of $E_i$ and~$E_j$ is in~$\bbE_1$ and the other is in~$\bbE_2$ (line~\ref{line:alphaUJ_addEdges}) and, hence, $E_i$ and~$E_j$ are on different sides of~$S$ in~$T$ (line~\ref{line:alphaUJ_partition}).
This implies that $S \subseteq S_i$ and $S \subseteq S_j$ (line~\ref{line:alphaUJ_detSubsets} and line~\ref{line:alphaUJ_detEdges}).
Therefore, by Lemma~\ref{lem:unionJoinProperties}, there is a join tree for~$H$ with the edge~$E_iE_j$.
\qed
\end{proof}

\begin{proof}[Complexity]
Creating a join tree for a given acyclic hypergraph~$H$ can be implemented in $\calO(N)$ time~\cite{TarjanYannak1984}.
Modifying that join tree (thereby computing~$T$) and computing $\calS(H)$ using Algorithm~\ref{algo:modifyT} can also be done in $\calO(N)$ time (Lemma~\ref{lem:modifyTRuntime}).
Thus, line~\ref{line:alphaUJ_preprocc} runs in total $\calO(N)$ time.
Computing the subset graph~$G_\calS$ in line~\ref{line:alphaUJ_callA} requires $\calO \bigl( T_\calA(H) \bigr)$ time.
Since the hyperedges of~$H$ form the vertices of~$G$ and since $G$ is created without edges, line~\ref{line:alphaUJ_createG} runs in $\calO(m)$ time.

We show next that a single iteration of the loop starting in line~\ref{line:alphaUJ_sepLoop} runs in $\calO \bigl( |\bbE_1| \cdot |\bbE_2| \bigr)$ time.
That is, the runtime for a single iteration is (asymptotically) equivalent to the number of edges of~$G$ created.
Note that each iteration creates at least one such edge, namely the edge in~$T$ that $S$ represents.
Additionally, Lemma~\ref{lem:unionJoinProperties} and Lemma~\ref{lem:atMost2Sep} imply that each edge~$E_iE_j$ is added at most twice to~$G$.
Therefore, line~\ref{line:alphaUJ_sepLoop} to line~\ref{line:alphaUJ_addEdges} run in $\calO \bigl( |G| \bigr)$ total time.

For a separator $S \in \calS$, let $\bbS$ denote the set of separators~$S'$ with $S \subseteq S'$.
Since the subset graph~$G_\calS$ is given, one can compute $\bbS$ (line~\ref{line:alphaUJ_detSubsets}) in $\calO \bigl( |\bbS| \bigr)$ time by determining all incoming edges of~$S$ in~$G_\calS$.
For each $S' \in \bbS$, the algorithm adds, in line~\ref{line:alphaUJ_detEdges}, exactly one hyperedge into~$\bbE$ plus one additional hyperedge for~$S$.
Thus, $|\bbE| = |\bbS| + 1$.

One can determine the hyperedges $E$ and~$E'$ that form a separator~$S'$, which one is farther from~$S$, and on which side of~$S$ they are in~$T$ as follows.
When creating~$S'$, add a reference to both hyperedges and include which is the parent and which is the child in~$T$.
Now assume that each $S'$ is also a node of~$T$ adjacent to $E$ and~$E'$.
Root $T$ in an arbitrary hyperedge, run a pre-order and post-order on~$T$, and let $\mathrm{pre}(x)$ and~$\mathrm{post}(x)$ be the indices of a node~$x$ in that respective order.
For two distinct nodes $x$ and~$y$ of~$T$ (representing either separators or hyperedges), $x$ is then a descendant of~$y$ if and only if $\mathrm{pre}(x) > \mathrm{pre}(y)$ and $\mathrm{post}(x) < \mathrm{post}(y)$.
There are four cases when determining which of $E$ and~$E'$ to add into~$\bbE$:
if $S$ and $S'$ represent the same edge of~$T$, add both hyperedges;
if $S'$ is a descendant of~$S$, add the child-hyperedge;
if $S'$ is an ancestor of~$S$, add the parent-hyperedge; and
if $S'$ is neither an ancestor nor a descendant of~$S$, add the child-hyperedge.
Clearly, one side of~$S$ contains all its descendants and the other side all remaining hyperedges and separators.
That allows us, after a $\calO(m)$-time preprocessing, to determine in constant time on which side of~$S$ a give a hyperedge is.
Therefore, line~\ref{line:alphaUJ_detEdges} and line~\ref{line:alphaUJ_partition} run in $\calO \bigl( |\bbE| \bigr)$ time.

Line~\ref{line:alphaUJ_addEdges} clearly runs in $\calO \bigl( |\bbE_1| \cdot |\bbE_2| \bigr)$ time.
Recall that $|\bbS| + 1 = |\bbE| = |\bbE_1| + |\bbE_2|$.
Therefore, a single iteration of the loop starting in line~\ref{line:alphaUJ_sepLoop} also runs in $\calO \bigl( |\bbE_1| \cdot |\bbE_2| \bigr)$ time.
\qed
\end{proof}

Recall that there is an algorithm which computes the subset graph for any given hypergraph in $\calO \bigl( N^2 / \log N \bigr)$ time~\cite{Pritchard1999}.
Thus, we have the following.

\begin{theorem}
There is an algorithm that computes the union join graph~\( G \) of an acyclic hypergraph in $\calO \bigl( N^2 / \log N + |G| \bigr)$ time.
\end{theorem}

The upper bound of at most $\Theta\bigl( N^2 / \log^2 N \bigr)$ many edges for any subset graph~\cite{Pritchard1999} does not apply to union join graphs.
Consider a hypergraph~$H = (V, \calE)$ with $V = \{ u, v_1, \ldots, v_n \}$ and $\calE = \{ \, E_i \mid 1 \leq i \leq n \, \}$ where $E_i = \{ u, v_i \}$.
Note that $N = 2n$ and that each tree with $\calE$ as nodes is a valid join tree for~$H$.
Hence, the union join graph of~$H$ is a complete graph with $\Theta \bigl( N^2 \bigr)$ edges.

\subsection{Notes on the Sperner Family Problem and its Generalisation}

An interesting question that remains is the complexity of solving the Sperner Family problem for acyclic hypergraphs and hypertrees.
We first answer this question for hypertrees.

\begin{theorem}
If the SETH is true, then there is no algorithm which solves the Sperner Family problem for a given hypertree in \( \calO \big( N^{2 - \varepsilon} \big) \) time.
\end{theorem}

\begin{proof}
We prove the theorem by making a simple linear-time reduction.
Consider a family~$\calF = \{ S_1, S_2, \ldots, S_m \}$ of sets.
Create a new vertex~$u$ and add it to each set~$S_i \in \calF$.
Let $S'_i$ be the resulting set.
The family~$\calF' = \{ S'_1, S'_2, \ldots, S'_m \}$ then forms a hypertree.
Clearly, adding $u$ to each set does not change any subset relations.
Therefore, $\calF$ contains two distinct sets $S_i$ and~$S_j$ with $S_i \subseteq S_j$ if and only if $\calF'$ contains two distingt sets $S'_i$ and~$S'_j$ with $S'_i \subseteq S'_j$.
\qed
\end{proof}

For acyclic hypergraphs we have the following result.

\begin{theorem}
\label{theo:alphaSFlinear}
There is a linear-time algorithm to solve the Sperner Family problem for acyclic hypergraphs.
\end{theorem}

\begin{proof}
Let $H$ be an acyclic hypergraph with a join tree~$T$.
We first show that the following are equivalent:
(i)~$H$ has two distinct hyperedges $E_i$ and~$E_j$ with $E_i \subseteq E_j$, and
(ii)~$T$ has an edge $EE'$ with $E \subseteq E'$.
Clearly, (ii) implies (i).
To show that (i) implies (ii), let $E_i$ and~$E_j$ be two distinct hyperedges of~$H$ with $E_i \subseteq E_j$.
It follows from the definition of join trees that $E_i \subseteq E_k$ for each hyperedge~$E_k$ on the path from $E_i$ to~$E_j$ in~$T$.
Therefore, $T$ contains an edge $E_iE_k$ with $E_i \subseteq E_k$ and, thus, (i) is equivalent to~(ii).

Let $EE'$ be an edge of~$T$ with $E \subseteq E'$, and let $S = E \cap E'$ be the separator corresponding to that edge.
Note that $|S| = |E|$ in such a case.
Hence, it follows that (i) is true if and only if $H$ admits a separator~$S = E \cap E'$ such that $|S| = |E|$ or $|S| = |E'|$.

We can now solve the Sperner Family problem for a given acyclic hypergraph~$H$ in linear time as follows.
Construct the separator hypergraph for~$H$ (see Algorithm~\ref{algo:modifyT}).
For each separator~$S = E \cap E'$, determine if $|S| = |E|$ or $|S| = |E'|$.
In that case, return \emph{True}.
Otherwise, if no such separator is found, return \emph{False}.
\qed
\end{proof}

Theorem~\ref{theo:hardnessSubsetGraph} and Theorem~\ref{theo:alphaSFlinear} together give an interesting observation:
Let $\calC$ he a class of hypergraphs.
Existence of an algorithm that solves the Sperner Family Problem for~$\calC$ in truly subquadratic time does not imply that there is such an algorithm to compute a subset graph for hypergraphs in~$\calC$, even if the resulting graph is sparse.

One can generalise the Sperner Family Problem as follows:
How many pairs of distinct sets $S_i, S_j$ with $S_i \subseteq S_j$ does a given a family~$\calF$ contain?
Let $p$ be that number.
The Sperner Family problem is then equal to the question whether or not $p \geq 1$.
Thus, Theorem~\ref{theo:alphaSFlinear} gives linear-time algorithm to determine if $p \geq 1$ for acyclic hypergraphs.
The reduction for Lemma~\ref{lem:hardnessUniqueJoinTree}, however, implies that there is no truly subquadratic-time algorithm that determines if $p \geq m$.
What remains an open question is the required runtime to determine if $p \geq k$ for any fixed~$k$ with $1 < k < m$.

\section{$\beta$-Acyclic Hypergraphs}
\label{sec:betaAcyclic}

A hypergraph~$H = (V, \calE)$ is \emph{\( \beta \)-acyclic} if each subset of~$\calE$ forms an acyclic hypergraph.
They are also known as \emph{totally balanced} hypergraphs~\cite{DAtriMoscar1988b}.
See~\cite{Fagin1983} for more definitions.
In this section, we present an algorithm to compute the subset graph~$G$ of $\beta$-acyclic hypergraphs in $\calO \bigl( N \log (n + m) + |G| \bigr)$ time.
Afterwards, we show that one can use that algorithm together with Algorithm~\ref{algo:alphaUnionJoin} to compute the union join graph in the same amount of time.

\subsection{Constructing the Subset Graph}

A matrix is \emph{binary} if its entries are either $0$ or~$1$.
The binary matrix~%
$
    \bigl[
    \begin{smallmatrix}
        1 & 1 \\
        1 & 0
    \end{smallmatrix}
    \bigr]
$
is called~\emph{\( \Gamma \)}.
A matrix is \emph{\( \Gamma \)-free} if it contains no $\Gamma$ as submatrix.
Note that the rows and columns which form a~$\Gamma$ submatrix do not need to be adjacent in the original matrix.
One can use a binary $n \times m$ matrix~$M$ to represent a given hypergraph~$H = (V, \calE)$ as follows.
Let each row~$i$ represent a vertex~$v_i \in V$ and each column~$j$ represent a hyperedge~$E_j \in \calE$.
An entry~$M_{i, j}$ is then~$1$ if and only if $v_i \in E_j$.
That matrix is called the \emph{incidence matrix} of~$H$.

A matrix is \emph{doubly lexically ordered} if rows and columns are permuted in such a way that rows vectors and columns are both in non-decreasing lexicographic order (rows from left to right and columns from top to bottom).
Within a row, priorities of entries are decreasing from right to left, and, within a column, priorities of entries are decreasing from bottom to top.
One can compute such an ordering in $\calO \bigl( N \log (n + m) \bigr)$ time~\cite{PaigeTarjan1987}.
Note that the algorithm in~\cite{PaigeTarjan1987} does not compute the actual matrix; it only computes the corresponding ordering of vertices and hyperedges, thereby avoiding a quadratic runtime.

\begin{lemma}
\label{lem:betaGammaFree}
\cite{Berge1989,DAtriMoscar1988b}
A hypergraph is \( \beta \)-acyclic if and only if its doubly lexically ordered incidence matrix is \( \Gamma \)-free.
\end{lemma}

For the remainder of this subsection, assume that we are given a $\beta$-acyclic hypergraph~$H = (V, \calE)$.
Let $M$ be a doubly lexically ordered (hence, $\Gamma$-free) incidence matrix for~$H$.
We assume that we know the ordering of vertices and hyperedges in~$M$, even though we are not given $M$ itself.
For two hyperedges $E_i$ and~$E_j$ of~$H$, we say $E_i \preceq E_j$ if the column of~$E_i$ is lexicographically smaller than or equal to the column of~$E_j$ with respect to~$M$.
Accordingly, we write $E_i \prec E_j$ to exclude equality.

\begin{lemma}
\label{lem:betaSubsetCondition}
Let \( E_i \) and~\( E_j \) be two hyperedges of~\( H \) and let \( v \) be the vertex in~\( E_i \) which is earliest in the doubly lexical ordering (\ie, highest in~\( M \)).
Then, \( E_i \subseteq E_j \) if and only if \( E_i \preceq E_j \) and \( v \in E_j \).
\end{lemma}

\begin{proof}
We first show that $E_i \preceq E_j$ and $v \in E_j$ implies $E_i \subseteq E_j$.
Clearly, $E_i \subseteq E_j$ if $E_i = E_j$.
Assume now that $E_i \nsubseteq E_j$, \ie, $E_i$ contains a vertex~$u \notin E_j$.
By definition of~$v$, $u$ is lower in~$M$ than~$v$.
$E_i \preceq E_j$ and $v \in E_j$ then imply that $E_i$, $E_j$, $u$, and $v$ form a~$\Gamma$ in~$M$.
That contradicts with $M$ being $\Gamma$-free (see Lemma~\ref{lem:betaGammaFree}).
Therefore, $E_i \succ E_j$ or $v \notin E_j$.

Clearly, $v \notin E_j$ implies $E_i \nsubseteq E_j$.
Now assume that $E_i \succ E_j$.
Since $E_i \neq E_j$, there is a lowest vertex~$u$ in~$M$ which is in one of these hyperedges but not in both.
Recall that $M$ is ordered lexicographically.
Therefore, $E_i \succ E_j$ implies that $u \in E_i$ ($1$~in~$M$) and $u \notin E_j$ ($0$~in~$M$), \ie, $E_i \nsubseteq E_j$.
\qed
\end{proof}

Lemma~\ref{lem:betaSubsetCondition} allows to compute the subset graph~$G$ of a $\beta$-acyclic hypergraph as follows.
First, find doubly lexicographical ordering of vertices and hyperedges.
For each hyperedge~$E$, determine all hyperedges~$E'$ with $E \preceq E'$ which contain~$v$ as defined in Lemma~\ref{lem:betaSubsetCondition}.
Then, add the edge~$(E, E')$ to~$G$ for each such pair $E$ and~$E'$.
Algorithm~\ref{algo:betaSubset} implements that approach.

\begin{algorithm}[!htb]
\caption
{%
    Computes the subset graph of a given $\beta$-acyclic hypergraph.
}
\label{algo:betaSubset}

\KwIn
{%
    A $\beta$-acyclic hypergraph~$H = (V, \calE)$.
}

\KwOut
{%
    The subset graph~$G$ of~$H$.
}

Find doubly lexicographical ordering~$\sigma$ of vertices and hyperedges (see~\cite{PaigeTarjan1987}) and order the adjacency list of~$\calI(H)$ according to~$\sigma$.
\label{line:betaSubset_findOrdering}

Create a new directed graph~$G = (\calE, E_G)$ with $E_G = \emptyset$.
\label{line:betaSubset_initG}

\ForEach
{%
    \( E \in \calE \)%
    \label{line:betaSubset_outerLoop}
}
{%
    Let $v$ be the vertex in~$E$ which is first in~$\sigma$.
    \label{line:betaSubset_pickV}

    \ForEach
    (%
        {\hfill (\( \preceq \) with respect to~\( \sigma \))}%
    )
    {%
        hyperedge~\( E' \) containing~\( v \) with \( E \preceq E' \)%
        \label{line:betaSubset_innerLoop}
    }
    {%
        Add $(E, E')$ to $E_G$.
        \label{line:betaSubset_addEdge}
    }
}

\end{algorithm}

\begin{theorem}
Algorithm~\ref{algo:betaSubset} computes the subset graph~\( G \) of a given \( \beta \)-acyclic hypergraph in \( \calO \bigl( N \log (n + m) + |G| \bigr) \) time.
\end{theorem}

\begin{proof}[Correctness]
We show the correctness of Algorithm~\ref{algo:betaSubset} by showing that $G$ contains an edge~$(E_i, E_j)$ if and only if $E_i \subseteq E_j$.
First assume that $G$ contains an edge~$(E_i, E_j)$.
Note that Algorithm~\ref{algo:betaSubset} only adds $(E_i, E_j)$ to~$G$ (line~\ref{line:betaSubset_addEdge}) if $E_i \preceq E_j$ and $v \in E_j$ (line~\ref{line:betaSubset_innerLoop}).
Therefore, due to Lemma~\ref{lem:betaSubsetCondition}, $G$ containing an edge~$(E_i, E_j)$ implies $E_i \subseteq E_j$.

Next, assume that $H$ contains two hyperedges $E_i$ and~$E_j$ with $E_i \subseteq E_j$.
It then follows from Lemma~\ref{lem:betaSubsetCondition} that $E_i \preceq E_j$ and $v \in E_j$.
Since the algorithm checks all pairs of hyperedges satisfying this condition (line~\ref{line:betaSubset_outerLoop} and line~\ref{line:betaSubset_innerLoop}), it eventually finds $E_i$ and~$E_j$ and adds $(E_i, E_j)$ as edge to~$G$ (line~\ref{line:betaSubset_addEdge}).
\qed
\end{proof}

\begin{proof}[Complexity]
One can compute a doubly lexicographical ordering~$\sigma$ (line~\ref{line:betaSubset_findOrdering}) in $\calO \bigl( N \log (n + m) \bigr)$ time~\cite{PaigeTarjan1987}.
Creating the graph~$G$ (line~\ref{line:betaSubset_initG}) can easily be done in $\calO(m)$ time.

There are various ways to then order the adjacency list of~$\calI(H)$ (line~\ref{line:betaSubset_findOrdering}) in $\calO(N)$ time.
One option is to reconstruct $\calI(H)$ as follows.
Iterate over all vertices~$v$ of~$H$ as ordered in~$\sigma$.
For each hyperedge~$E$ containing~$v$, add~$v$ to the new list of~$E$.
Afterwards, the list of vertices in~$E$ is ordered with respect to~$\sigma$.
The same approach (with hyperedges and vertices swapped) allows to sort, for each vertex~$v$, the list of hyperedges containing it.

We now show that the loop starting in line~\ref{line:betaSubset_outerLoop} runs is in $\calO \bigl( |G| \bigr)$ total time.
Note that the hyperedges~$\calE$ form the vertices of~$G$.
Hence, there is exactly one iteration of the loop starting in line~\ref{line:betaSubset_outerLoop} for each vertex of~$G$.
Since the adjacency list of~$\calI(H)$ is ordered according to~$\sigma$ (line~\ref{line:betaSubset_findOrdering}), we can determine $v$ (line~\ref{line:betaSubset_pickV}) in constant time for each hyperedge~$E$.
For the same reason, we can determine the set~$\bbE_v = \{ \, E' \in \calE \mid v \in E', E \preceq E' \, \}$ in $\calO \bigl( |\bbE_v| \bigr)$ time by iterating backwards over the hyperedges containing~$v$.
Since we add exactly one edge~$(E, E')$ to~$G$ for each~$E'$ in such~$\bbE_v$, line~\ref{line:betaSubset_innerLoop} and line~\ref{line:betaSubset_addEdge} run in $\calO \bigl( |G| \bigr)$ total time.
\qed
\end{proof}

\subsection{Constructing the Union Join Graph}

We now address how to compute the union join graph for $\beta$-acyclic hypergraphs.
For that, we do not present a new algorithm.
Instead, we show that one can use Algorithm~\ref{algo:alphaUnionJoin} together with Algorithm~\ref{algo:betaSubset}.
This is possible due to Lemma~\ref{lem:betaSeparator} below.

\begin{lemma}
\label{lem:betaSeparator}
If a hypergraph is \( \beta \)-acyclic, then its separator hypergraph is \( \beta \)-acyclic, too.
\end{lemma}

Before proving Lemma~\ref{lem:betaSeparator}, we need a few auxiliary definitions.
Assume we are given a graph~$G = (V, E)$.
A \emph{clique} is a set of vertices of~$G$ such that all these vertices are pairwise adjacent.
Such a clique~$K$ is \emph{maximal} if no vertex in~$G$ is adjacent to all vertices in~$K$.
For a cycle $\langle v_1, v_2, \ldots, v_k, v_1 \rangle$, a \emph{chord} is an edge between two non-consecutive vertices.
A graph is called \emph{chordal} if each cycle with four or more vertices has a chord.
A hypergraph~$H$ is \emph{conformal} if, for each maximal clique~$K$ of~$\2Sec(H)$, $H$ contains a hyperedge~$E$ with $K \subseteq E$.

\begin{proof}[Lemma~\ref{lem:betaSeparator}]
Let $H = (V, \calE)$ be a $\beta$-acyclic hypergraph with a join tree~$T$ and let $\calS$ be the hyperedges of~$\calS(H)$.
To prove that $\calS(H)$ is $\beta$-acyclic, we show that each $\calS' \subseteq \calS$ forms an acyclic hypergraph.
It is known that a hypergraph~$H$ is acyclic if and only if $H$ is conformal and $\2Sec(H)$ is chordal~\cite{BeeFagMaiYan1983}.
It is therefore sufficient for us to show that each $\calS'$ is conformal and its 2-section graph is chordal.

We start by showing that $G_2 = \2Sec \bigl( \calS' \bigr)$ is chordal.
Let us assume that $G_2$ is not chordal.
It then contains a chordless cycle~$C = \{ v_1, v_2, \ldots, v_k \}$ with $k \geq 4$.
Since each edge~$v_iv_{i + 1}$ of~$G_2$ implies that its vertices are in a common hyperedge, there is a sequence of separators $\sigma = \langle S_1, S_2, \ldots S_k \rangle$ that form~$C$ in~$G_2$.
In particular, we have that $v_i \in S_i \cap S_{i + 1}$ (with index arithmetic modulo~$k$).
Recall that each separator~$S_i$ corresponds to an edge of~$T$.
Let $T_\sigma$ be the smallest subtree of~$T$ that contains all separators of~$\sigma$.
Note that $v_i \notin S_j$ for all $j \notin \{ i, i + 1 \}$; otherwise $C$ would have a chord.
Therefore, by properties of join trees, there is no $i$ and~$j$ such that $S_j$ separates $S_i$ and~$S_{i + 1}$ in~$T$.
Hence, each $S_i$ in~$\sigma$ corresponds to a leaf~$E_i$ of~$T_\sigma$.
By properties of join trees, $E_i \cap E_j = S_i \cap S_j$.
Hence, $\2Sec \bigl( \{ E_1, E_2, \ldots, E_k \} \bigr)$ is not chordal implying that $\{ E_1, E_2, \ldots, E_k \}$ does not form an acyclic hypergraph.
This contradicts with $H$ being $\beta$-acyclic.
Therefore, $\2Sec \bigl( \calS' \bigr)$ is chordal.

We now show that each $\calS'$ forms a conformal hypergraph.
\emph{Gilmore's Theorem}~\cite{Berge1989} states that a hypergraph~$H$ is conformal if and only if, for all its hyperedges $E_1$, $E_2$, and~$E_3$, $H$ contains a hyperedge~$E$ such that $(E_1 \cap E_2) \cup (E_2 \cap E_3) \cup (E_1 \cap E_3) \subseteq E$.
$\calS'$ is therefore clearly conformal if $|\calS'| \leq 2$.
Now let $|\calS'| \geq 3$ and let $S_1$, $S_2$, and~$S_3$ be three hyperedges in~$\calS'$.
We distinguish between two cases.
Case~1:~$S_1$, $S_2$, and~$S_3$ are on a path in $T$.
Without loss of generality, let $S_2$ be between $S_1$ and~$S_3$.
Then, by properties of join trees, $S_1 \cap S_3 \subseteq S_2$.
Case~2:~There is a hyperedge~$E$ in~$H$ such that $S_1$, $S_2$, and~$S_3$ are in different subtrees when removing~$E$ from~$T$.
For all $i \in \{ 1, 2, 3 \}$, let $S_i$ represent the edge~$E_iE'_i$ of~$T$ and let $E'_i$ be closer to~$E$ in~$T$ than~$E_i$.
Since $H$ is $\beta$-acyclic, the set $\{ E_1, E_2, E_3 \}$ also forms an acyclic and, hence, conformal hypergraph.
Without loss of generality, let $E_3$ be the hyperedge that satisfies \emph{Gilmore's Theorem} for $\{ E_1, E_2, E_3 \}$, \ie, let $E_1 \cap E_2 \subseteq E_3$.
Note that, by properties of join trees, $E_i \cap E_j = S_i \cap S_j$ for all $i, j \in \{ 1, 2, 3 \}$, and $v \in S_3$ if $v \in (E_1 \cup E_2) \cap E_3$.
Therefore, $S_1 \cap S_2 \subseteq S_3$.
Overall, it then follows each $\calS'$ forms a conformal hypergraph.
\qed
\end{proof}

Due to Lemma~\ref{lem:betaSeparator}, we can conclude this section as follows.

\begin{theorem}
\label{theo:betaUnionJoin}
There is an algorithm that computes the union join graph~\( G \) of a given \( \beta \)-acyclic hypergraph in \( \calO \bigl( N \log (n + m) + |G| \bigr) \) time.
\end{theorem}

\begin{proof}
Let $H$ be the given hypergraph.
Because the separator hypergraph~$\calS(H)$ is $\beta$-acyclic (Lemma~\ref{lem:betaSeparator}), we can use Algorithm~\ref{algo:betaSubset} to compute its subset graph~$G'$ in $\calO \bigl( N \log (n + m) + |G'| \bigr)$ time.
Thus, when using $H$ and Algorithm~\ref{algo:betaSubset} as input, Algorithm~\ref{algo:alphaUnionJoin} computes the union join graph~$G$ of~$H$ in $\calO \bigl( N \log (n + m) + |G'| + |G| \bigr)$ time.
Consider again line~\ref{line:alphaUJ_sepLoop} to line~\ref{line:alphaUJ_addEdges} of Algorithm~\ref{algo:alphaUnionJoin}.
Note that for each edge~$(S, S')$ of~$G'$, there is at least one edge added to~$G$.
It follows that $|G'| \leq |G|$ and, therefore, one can compute~$G$ in $\calO \bigl( N \log (n + m) + |G| \bigr)$ total time.
\qed
\end{proof}

\section{$\gamma$-Acyclic Hypergraphs}
\label{sec:gammaAcyclic}

In~\cite{Fagin1983}, \textsc{Fagin} gives various definitions of $\gamma$-acyclic hypergraphs and presents a polynomial-time recognition algorithm for them.
The definition for $\gamma$-acyclic hypergraphs we give below uses a strong relation between these hypergraphs and distance-hereditary graphs.
Before that, we give a few auxiliary definitions and an interesting property of distance-hereditary graphs.

Let $G = (V, E)$ be a connected, undirected, and simple graph without loops or multiple edges.
The \emph{open} and \emph{closed neighbourhood} of a vertex~$v \in V$ are respectively defined as $N(v) = \{ \, u \mid uv \in E \, \}$ and $N[v] = N(v) \cup \{ v \}$.
A vertex~$v$ is \emph{pendant} if $\bigl| N(v) \bigr| = 1$.
Two vertices $u$ and~$v$ are \emph{false twins} if $N(u) = N(v)$, and are \emph{true twins} if $N[u] = N[v]$.
A graph $G$ is \emph{distance-hereditary} if every induced subgraph is distance preserving, \ie, the distance between two vertices $u$ and~$v$ remains the same in every connected induced subgraph of~$G$ that contains $u$ and~$v$.

An ordering~$\sigma = \langle v_1, v_2, \ldots, v_n \rangle$ for a graph~$G$ is called a \emph{pruning sequence} if each~$v_i$ with $i > 1$ satisfies one of the following conditions in the subraph of~$G$ induced by~$\{ v_1, v_2, \ldots, v_i \}$:
(i)~$v_i$ is pendant,
(ii)~$v_i$ is a true twin of some vertex~$v_j$, or
(iii)~$v_i$ is a false twin of some vertex~$v_j$.
A graph~$G$ is distance-hereditary if and only if $G$ admits a pruning sequence~\cite{BandelMulder1986}.

The recognition algorithm in~\cite{Fagin1983} decides whether or not a given hypergraph is $\gamma$-acyclic by determining if the corresponding incidence graph admits a pruning sequence.
Additionally, \textsc{Ausiello}~et~al.\,\cite{AusiDAtrMosc1986} show that the incidence graphs of $\gamma$-acyclic hypergraphs are so-called \emph{(6, 2)-chordal bipartite} graph which are known to be equivalent to bipartite distance-hereditary graphs~\cite{DAtriMoscar1988a}.
Therefore, we can define $\gamma$-acyclic hypergraphs as follows.

\begin{corollary}
\label{cor:gammaDef}
\cite{AusiDAtrMosc1986,DAtriMoscar1988a}\,%
\cite{Fagin1983}
A hypergraph is \emph{\( \gamma \)-acyclic} if and only if its incidence graph is distance-hereditary.
\end{corollary}

\subsection{Constructing the Union Join Graph}

\begin{lemma}
\label{lem:gammaUnionJoinLine}
An acyclic hypergraph is \( \gamma \)-acyclic if and only if its line graph is isomorphic to its union join graph.
\end{lemma}

\begin{proof}
Let $H$ be an acyclic hypergraph with two distinct hyperedges $E_i$ and~$E_j$, and let $G$ be the union join graph of~$H$.
Consider the following statements:
(i)~$E_iE_j$ is an edge of~$L(H)$,
(ii)~$E_i \cap E_j \neq \emptyset$,
(iii)~$E_i \cap E_j$ separates $E_i \setminus E_j$ from $E_j \setminus E_i$, and
(iv)~$E_iE_j$ is an edge of~$G$.
Due to definitions and due to Lemma~\ref{lem:unionJoinProperties}, it follows that (i) is equivalent to~(ii), that (iii) implies~(ii), and that (iii) is equivalent to~(iv).

To prove Lemma~\ref{lem:gammaUnionJoinLine}, we first assume that $H$ is $\gamma$-acyclic.
It is know~\cite{Fagin1983} that a hypergraph is $\gamma$-acyclic if and only if (ii) implies~(iii) for all distinct hyperedges $E_i$ and~$E_j$.
Therefore, if $H$ is $\gamma$-acyclic, the statements (i), (ii), (iii), and~(iv) are equivalent and, subsequently, $L(H) = G$.

Now assume that $L(H) = G$, \ie, that (i) and~(iv) are equivalent for all distinct hyperedges $E_i$ and~$E_j$.
It then follows that (ii) and~(iii) are equivalent and, as a result, that (ii) implies~(iii).
The same observation from~\cite{Fagin1983} then implies that $H$ is $\gamma$-acyclic if $L(H) = G$.
\qed
\end{proof}

\begin{theorem}
There is an algorithm that computes the union join graph~\( G \) of a given \( \gamma \)-acyclic hypergraph in \( \calO \bigl( N + |G| \bigr) \) time.
\end{theorem}

\begin{proof}
Due to Lemma~\ref{lem:gammaUnionJoinLine}, we can compute the union join graph~$G$ of a $\gamma$-acyclic hypergraph~$H$ by computing its line graph.
Note that, by definition, $L(H) = \2Sec(H^*)$.
It follows from Corollary~\ref{cor:gammaDef} that the dual hypergraph~$H^*$ is $\gamma$-acyclic too.
Therefore, we can compute $G = \2Sec(H^*)$ as follows.

Let $T$ be a join tree of~$H^*$ rooted in an arbitrary node.
Process each hyperedge of~$H^*$ according to a pre-order on~$T$.
When processing a hyperedge~$E$ of~$H^*$, pick a vertex~$v \in E$ that has not been flagged, make $v$ adjacent (in~$G$) to all flagged vertices in~$E$, and then flag~$v$.
Repeat that until all vertices in~$E$ are flagged and afterwards continue with the next hyperedge.

By flagging vertices, we ensure that an edge~$uv$ is added to~$G$ at most once even if both vertices are together in multiple hyperedges of~$H^*$.
Therefore, since we can construct~$T$ in $\calO(N)$ time~\cite{TarjanYannak1984}, we can construct~$G$ from~$H$ in $\calO \bigl( N + |G| \bigr)$ time.
\qed
\end{proof}

\subsection{Constructing the Subset Graph}

\subsubsection{Bachman Diagrams.}

Consider a hypergraph~$H = (V, \calE)$, let $\calE'$ be a subset of~$\calE$, and let $\frkX$ be the intersection of all hyperedges in~$\calE'$.
We then define $\calX$ as the set of all such $\frkX$ which are non-empty, \ie,
\[
    \calX = \bigcup_{\calE' \subseteq \calE} \Bigl \{ \, \frkX \Bigm| {\textstyle \frkX = \bigcap_{E \in \calE'} E}, \frkX \neq \emptyset \, \Bigr \}.
\]
The \emph{Bachman diagram \( \calB(H) \)} of~$H$ is a directed graph with the node set~$\calX$ such that there is an edge from $\frkX$ to~$\frkY$ if $\frkX \supset \frkY$ and there is no $\frkZ$ with $\frkX \supset \frkZ \supset \frkY$.
Note that, if $H$ contains two distinct hyperedges $E_i$ and~$E_j$ with the same vertices, they are represented by the same node in~$\calB(H)$.

\begin{lemma}
\label{lem:bachmanTree}
\cite{Fagin1983}
A hypergraph is \( \gamma \)-acyclic if and only if its Bachman diagram forms a tree.
\end{lemma}

In a Bachman diagram~$\calB(H)$ as defined above, a vertex~$v$ of~$H$ is often contained in multiple nodes.
A technique from~\cite{UeharaUno2009} allows us to construct a more compact representation of~$\calB(H)$.
Let $N(\frkX)$ be the set of nodes~$\frkY$ such that $(\frkX, \frkY)$ is an edge of~$\calB(H)$.
We then define the \emph{label} of~$\frkX$ as $\ell(\frkX) := \frkX \setminus \bigcup_{\frkY \in N(\frkX)} \frkY$.
As a result, a vertex~$v$ of~$H$ is only in the label of the ``smallest'' node~$\frkY$ containing it.
Consider now a Bachman diagram~$\calB(H)$ with the node set~$\calX$ where we replace each node~$\frkX \in \calX$ with $X = \ell(\frkX)$.
We call the resulting graph~$B$ a \emph{simplified} Bachman diagram of~$H$.
Figure~\ref{fig:BachmanExample} gives an example.

\begin{figure}
    \centering

    \subfloat[][]
    {%
        \centering
        \begin{tikzpicture}
[
    xscale=2,
    yscale=1.5,
    every node/.style=
    {
        draw,
        inner sep = 5pt,
        minimum size = 1.5em,
    }
]

\node [rectangle] (nOr) at (0,  0) {$a, b, c$};
\node [rectangle] (nBl) at (0, -2) {$a, d$};
\node [rectangle] (nGr) at (1, -2) {$c, e, f$};
\node [rectangle] (nRe) at (1,  0) {$b, c$};

\node [circle] (nOB) at (0, -1) {$a$};
\node [circle] (nRG) at (1, -1) {$c$};

\begin{scope}
[
    -latex,
    shorten > = 1pt
]

\draw (nOr) -- (nRe);
\draw (nOr) -- (nOB);
\draw (nBl) -- (nOB);
\draw (nRe) -- (nRG);
\draw (nGr) -- (nRG);

\end{scope}
\end{tikzpicture}
        \label{fig:Bachman}
    }%
    \hfil%
    \subfloat[][]
    {%
        \centering
        \begin{tikzpicture}
[
    xscale=2,
    yscale=1.5,
    every node/.style=
    {
        draw,
        inner sep = 5pt,
        minimum size = 1.5em,
    }
]

\node [rectangle] (nOr) at (0,  0) {};
\node [rectangle] (nBl) at (0, -2) {$d$};
\node [rectangle] (nGr) at (1, -2) {$e, f$};
\node [rectangle] (nRe) at (1,  0) {$b$};

\node [circle] (nOB) at (0, -1) {$a$};
\node [circle] (nRG) at (1, -1) {$c$};

\begin{scope}
[
    -latex,
    shorten > = 1pt
]

\draw (nOr) -- (nRe);
\draw (nOr) -- (nOB);
\draw (nBl) -- (nOB);
\draw (nRe) -- (nRG);
\draw (nGr) -- (nRG);

\end{scope}
\end{tikzpicture}
        \label{fig:BachmanSimple}
    }%

    \caption
    {%
        The Bachman diagram~\protect\subref{fig:Bachman} and its simplified version~\protect\subref{fig:BachmanSimple} for a $\gamma$-acyclic hypergraph~$H$ with the hyperedges~$\{ a, b, c \}$, $\{ a, d \}$, $\{ b, c \}$, and $\{ c, e, f \}$.
        Nodes which represent hyperedges of~$H$ are drawn as rectangles; other nodes are drawn as circles.
    }
    \label{fig:BachmanExample}
\end{figure}
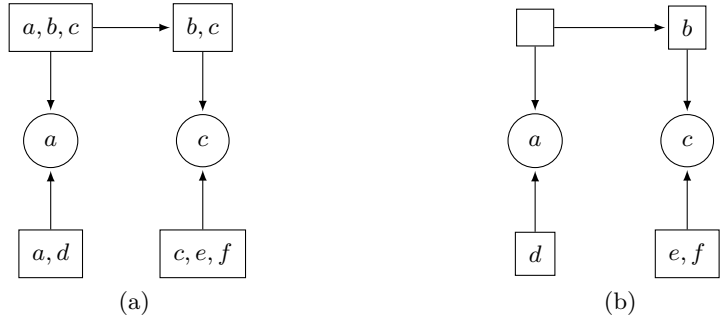

Let $B$ be a simplified Bachman diagram for a hypergraph~$H = (V, \calE)$.
We use the following functions and notations when working with $B$ and~$H$.
The function~$\phi$ maps~$\calE$ onto the nodes of~$B$ such that $\phi(E)$ returns the node which represents~$E$.
Accordingly, we define $\Phi(X) := \bigl \{ \, E \in \calE \bigm| \phi(E) = X \, \}$ for all nodes~$X$ of~$B$.
Similar to~$\phi$, we define $\psi$ as a function that maps $V$ onto the nodes of~$B$ such that $\psi(v)$ returns the node which contains~$v$.
For two nodes $X$ and~$Y$, we write $X \leadsto Y$ to state that there is a path form $X$ to~$Y$ in~$B$.
Note that we assume that $X \leadsto X$.
Lastly, we define $\bbV(X) = \{ \, v \in Y \mid X \leadsto Y \, \}$.
Note that $\bbV$ is effectively the inverse of the label function~$\ell$ we used above.

\subsubsection{Subset Graph via Simplified Bachman Diagrams.}

We can make the following observation:
For two hyperedges $E_i$ and~$E_j$ of~$H$, $E_i \subseteq E_j$ if and only if $\phi(E_j) \leadsto \phi(E_i)$ in the (simplified) Bachman diagram of~$H$.
In the remainder of this subsection, we present algorithms which first constructs a simplified Bachman diagram for a given $\gamma$-acyclic hypergraph~$H$ and then uses the previous observation to compute the subset graph~$G$ of~$H$ in $\calO \bigl( N + |G| \bigr)$ time.

To the best of our knowledge, there exist only two published algorithms which compute (simplified) Bachman diagrams.
\textsc{Kumar}~et~al.\,\cite{KumaShirGhos2009} present an $\calO \bigl( nm^2 \bigr)$-time algorithm to compute a Bachman diagram for a $\gamma$-acyclic database schema.
\textsc{Uehara} and \textsc{Uno}~\cite{UeharaUno2009} present a linear-time algorithm that computes a simplified Bachman diagram for the maximal cliques of a ptolemaic graph; these cliques form a $\gamma$-acyclic hypergraph~\cite{DAtriMoscar1988b}.
Using that algorithm would require us to first compute the 2-section graph of~$H$.
That may result in overall quadratic runtime for some hypergraphs.
We therefore use neither of these algorithms.
Instead, we present a new algorithm which computes a simplified Bachman diagram for a given $\gamma$-acyclic hypergraph in $\calO(N)$ time.

Recall that the incidence graph of a $\gamma$-acyclic hypergraph~$H$ is distance-hereditary.
It therefore admits a pruning sequence $\sigma = \langle x_1, x_2, \ldots, x_{n + m} \rangle$.
Note that each~$x_i$ in~$\sigma$ can represent either a vertex or a hyperedge of~$H$.
The idea for our algorithm is to iterate over~$\sigma$ and to step by step construct~$B$.
For that, let $\calI_i$ denote the subgraph of~$\calI(H)$ induced by~$\{ x_1, x_2, \ldots, x_i \}$.

We start the construction with $x_1$ and~$x_2$.
Note that one of them has to represent a vertex~$v$ and the other a hyperedge~$E$ of~$H$.
Therefore, we initialise~$B$ with a single node~$X = \{ v \}$ and set $\phi(E) := X$ and $\psi(v) := X$.

Next, we iterate over~$\sigma$, starting with~$x_3$.
Since incidence graphs are bipartite, it is never the case that $x_i$ is the true twin of some~$x_j$ (with the exception of~$i = 2$).
Hence, we have four possible cases for each $x_i$:
\begin{enumerate*}[(i)]
    \item
        \label{case:xTwinV}
        $x_i$ represents a vertex of~$H$ and is a false twin in~$\calI_i$,
    \item
        \label{case:xTwinE}
        $x_i$ represents a hyperedge and is a false twin,
    \item
        \label{case:xPendV}
        $x_i$ represents a vertex and is pendant, or
    \item
        \label{case:xPendE}
        $x_i$ represents a hyperedge and is pendant.
\end{enumerate*}

If $x_i$ is a twin (cases \ref{case:xTwinV} and~\ref{case:xTwinE}), the idea is to make the new vertex or hyperedge behave as its twin.
For a vertex~$v$, that means to add $v$ into the same node of~$B$.
In case of a hyperedge~$E$, it is represented by the same node of~$B$ as its twin.

If $x_i$ is pendant, adding it may affect the structure of~$B$.
For example, let $x_i$ represent a vertex~$v$ added to a hyperedge~$E$ (case~\ref{case:xPendV}).
If, with respect to $\calI_{i - 1}$, $E$ is not subset of another hyperedge (including not being a twin), then we can simply add $v$ into~$\phi(E)$.
However, if $E$ is subset of some hyperedge, it is no longer after adding~$v$.
We subsequently need to update the structure of~$B$.
To do so, we add a new node~$Y$, make it the representative of~$E$, and add an edge from~$Y$ to the node of~$B$ which previously represented~$E$.
We handle case~\ref{case:xPendE} in a similar way.

Algorithm~\ref{algo:bachman} implements the approach above and describes in detail how to handle each of the four cases for~$x_i$.

\begin{algorithm}[!htb]
\caption
{%
    Computes the Bachman diagram for a given $\gamma$-acyclic hypergraph.
}
\label{algo:bachman}

\KwIn
{
    A $\gamma$-acyclic hypergraph~$H = (V, \calE)$.
}

\KwOut
{
    A simplified Bachman diagram~$B$ for~$H$.
}

Compute a pruning sequence~$\sigma = \langle x_1, x_2, \ldots, x_{n + m} \rangle$ for~$\calI(H)$ (see~\cite{DamiHabiPaul2001}).
\label{line:bach_pruning}

Create a new empty graph~$B$.
\label{line:bach_emptyB}

Let $x_1$ and $x_2$ represent the vertex~$v$ and hyperedge~$E$ of~$H$.
Create a new set~$X = \{ v \}$, add it as node to~$B$, and set $\phi(E) := X$ and $\psi(v) := X$.
\label{line:bach_firstV}

\For
{%
    \( i := 3 \) \KwTo \( n + m \)%
    \label{line:bach_loop}
}
{%
    \If
    {%
        \( x_i \) represents a vertex~\( v \in V \) and is a false twin in~\( \calI_i \)
    }
    {%
        Let $u$ be the vertex represented by a twin of~$x_i$ in~$\calI_i$ and let $X = \psi(u)$.

        Add $v$ into~$X$, \ie, set $\psi(v) := X$ and $X := X \cup \{ v \}$.
        \label{line:bach_twinV_addV}
    }

    \smallskip

    \If
    {%
        \( x_i \) represents a hyperedge~\( E \in \calE \) and is a false twin in~\( \calI_i \)
    }
    {%
        Let $E'$ be the hyperedge represented by a twin of~$x_i$ in~$\calI_i$.

        Set $\phi(E) := \phi(E')$.
        \label{line:bach_twinE_setPhi}
    }

    \smallskip

    \If
    {%
        \( x_i \) represents a vertex~\( v \in V \) and is pendant in~\( \calI_i \)
    }
    {%
        Let $E$ be the hyperedge represented by the neighbour of~$x_i$ in~$\calI_i$ and let $X = \phi(E)$.

        \uIf
        {%
            \( \bigl| \Phi(X) \bigr| = 1 \) and \( X \) has no incoming edges in~\( B \)
        }
        {%
            Add $v$ into~$X$, \ie, set $\psi(v) := X$ and $X := X \cup \{ v \}$.
            \label{line:bach_pendV_addV}
        }
        \Else
        {%
            Create a new set~$Y = \{ v \}$, add it as node to~$B$, set $\psi(v) := Y$ and $\phi(E) := Y$, and add the edges~$(Y, X)$ to~$B$.
            \label{line:bach_pendV_newSet}
        }
    }

    \smallskip

    \If
    {%
        \( x_i \) represents a hyperedge~\( E \in \calE \) and is pendant in~\( \calI_i \)
    }
    {%
        Let $v$ be the vertex represented by the neighbour of~$x_i$ in~$\calI_i$ and let $X = \psi(v)$.

        \uIf
        {%
            \( |X| = 1 \) and \( X \) has no outgoing edges in~$B$
        }
        {%
            Set $\phi(E) := X$.
            \label{line:bach_pendE_setPhi}
        }
        \Else
        {%
            Create a new set~$Y = \{ v \}$, add it as node to~$B$, set $X := X \setminus \{ v \}$, set $\psi(v) := Y$ and $\phi(E) := Y$, and add the edge~$(X, Y)$ to~$B$.
            \label{line:bach_pendE_newSet}
        }
    }
}
\end{algorithm}

\begin{lemma}
\label{lem:bachmanAlgo}
Algorithm~\ref{algo:bachman} computes a simplified Bachman diagram for a given \( \gamma \)-acyclic hypergraph in linear time.
\end{lemma}

\begin{proof}[Correctness]
We start by showing that $B$ forms a tree.
Algorithm~\ref{algo:bachman} starts constructing~$B$ with a single node (line~\ref{line:bach_firstV}).
Whenever the algorithm adds a new node to~$B$ (line~\ref{line:bach_pendV_newSet} and line~\ref{line:bach_pendE_newSet}), it is incident to exactly one edge.
Additionally, no other edge is ever added to or removed from~$B$.
Therefore, $B$ forms a tree.

To show that $B$ is a simplified Bachman diagram for~$H$, we show that it satisfies the following two properties:
\begin{enumerate}[(1)]
    \item
        \label{prop:bachman_oneXforV}
        For each vertex~$v$ of~$H$, $B$ contains exactly one node~$X$ with $v \in X$; additionally, $\psi(v) = X$.
    \item
        \label{prop:bachman_bijection}
        There is a bijection~$f$ mapping $\calX$ onto the nodes of~$B$ such that
        \begin{enumerate*}[(a), ref=(\arabic{enumi}.\theenumii)]
            \item
                \label{prop:bachman_bij_f}
                $f(\frkX) = X$ if and only if $\frkX = \bbV(X)$, and
            \item
                \label{prop:bachman_bij_phi}
                $\frkX = E$ for some hyperedge~$E$ implies $f(\frkX) = \phi(E)$.
        \end{enumerate*}
\end{enumerate}
Property~\ref{prop:bachman_bijection} ensures that the nodes of~$B$ represent the nodes of a Bachman diagram for~$H$.
Property~\ref{prop:bachman_oneXforV} then enforces that the nodes of~$B$ are connected properly.
Without it, one could satisfy \ref{prop:bachman_bijection} by constructing a graph~$B = (\calX, \emptyset)$.
Additionally, since $B$ forms a tree, it does not contain transitive edges.

Observe that whenever a new vertex~$v$ is added (lines \ref{line:bach_firstV}, \ref{line:bach_twinV_addV}, \ref{line:bach_pendV_addV}, and~\ref{line:bach_pendV_newSet}), Algorithm~\ref{algo:bachman} adds it into a node~$X$ and sets~$\psi(v)$ accordingly.
In the case that an existing vertex~$v$ is added into a new node~$Y$ (line~\ref{line:bach_pendE_newSet}), the algorithms removes it from its previous node~$X$ and updates~$\psi(v)$ accordingly.
Therefore, the graph~$B$ constructed by Algorithm~\ref{algo:bachman} satisfies property~\ref{prop:bachman_oneXforV}.

In the remainder of this proof, we show that $B$ satisfies property~\ref{prop:bachman_bijection} via an induction over~$i$.
For that purpose, let $H_i$ denote the hypergraph formed by~$\calI_i$ and let $B_i$ be graph constructed after processing~$x_i$.
We also use subscript~$i$ to indicate that we refer to a version of a set, node, hyperedge, or function with respect to $B_i$ or~$H_i$; for larger expressions~$\varepsilon$, we may write $[ \varepsilon ]_i$.

Since $H_2$ has only one hyperedge and one vertex, $B_2$ (constructed in line~\ref{line:bach_firstV}) is clearly a simplified Bachman diagram for~$H_2$ and satisfies property~\ref{prop:bachman_bijection}.
In the following, we therefore assume that $i \geq 3$ and that $B_{i - 1}$ satisfies property~\ref{prop:bachman_bijection}.
We distinguish between four possible cases for~$x_i$.

\paragraph{Case~\ref{case:xTwinV}: \( x_i \) represents a vertex~\( v \in V \) and is a false twin in~\( \calI_i \).}
Let $u$ be the vertex represented by a twin of~$x_i$.
Since $u$ and~$v$ are twins, it follows that $v \in E_i$ if and only if $u \in E_{i - 1}$ for each hyperedge~$E$ of~$H_i$
Subsequently, the only change to~$\calX$ is that $v$ is added to the sets~$\frkX$ which contain~$u$.
That is, for each $\frkX \in \calX_i$, $\frkX_i = \frkX_{i - 1}$ if $u \notin \frkX_{i - 1}$ and $\frkX_i = \frkX_{i - 1} \cup \{ u \}$  if $u \in \frkX_{i - 1}$.
Observe that the algorithm neither adds any nodes nor any edges to the graph.
It only adds $v$ into~$\psi(u)_{i - 1}$.
Hence, for each node~$X$ of~$B_i$, $\bbV(X)_i = \bbV(X)_{i - 1}$ if $u \notin \bbV(X)_{i - 1}$ and $\bbV(X)_i = \bbV(X)_{i - 1} \cup \{ v \}$ if $u \in \bbV(X)_{i - 1}$.
Therefore, $B_i$ satisfies property~\ref{prop:bachman_bijection}.

\paragraph{Case~\ref{case:xTwinE}: \( x_i \) represents a hyperedge~\( E \in \calE \) and is a false twin in~\( \calI_i \).}
Recall that we defined Bachman diagrams and the family~$\calX$ in such a way that $\calX$ does not contain two equal sets, even if $H$ contains multiple equal hyperedges.
Hence, adding a hyperedge~$E$ which is equivalent to an existing hyperedge~$E'$ does neither change~$\calX$ nor any of the sets contained in it.
It follows that setting~$\phi(E)_i = \phi(E')_{i - 1}$ is the only change needed for~$B_i$ to satisfy property~\ref{prop:bachman_bijection} (otherwise $B_i$ would violate~\ref{prop:bachman_bij_phi}).
Algorithm~\ref{algo:bachman} does exactly that in line~\ref{line:bach_twinE_setPhi}.

\paragraph{Case~\ref{case:xPendV}: \( x_i \) represents a vertex~\( v \in V \) and is pendant in~\( \calI_i \).}
Let $E$ be the hyperedge represented by the neighbour of~$x_i$ in~$\calI_i$ (\ie,  $E_i = E_{i - 1}$), and let $X = \phi_{i - 1}(E)$.
Assume that, for each hyperedge~$E'$ of~$H_{i - 1}$ which is distinct from~$E$,  $E_{i - 1} \nsubseteq E'_{i - 1}$.
In that case, $\bigl| \Phi(X) \bigr|_{i - 1} = 1$ and $X$ has incoming edge in~$B_{i - 1}$.
(As result, Algorithm~\ref{algo:bachman} calls line~\ref{line:bach_pendV_addV}.)
Since $v$ is only added into~$E$, $\calX_i$ is almost identical to $\calX_{i - 1}$ except that the set~$\frkX$ which represents~$E$ now contains~$v$.
Because $X$ has incoming edge in~$B_{i - 1}$, adding $v$ into it (line~\ref{line:bach_pendV_addV}) does not affect other nodes.
In particular, $\bbV(Y)_i = \bbV(Y)_{i - 1}$ for all nodes~$Y$ of~$B_i$ which are distinct from~$X$, and $\bbV(X)_i = \bbV(X)_{i - 1} \cup \{ v \}$.
Therefore, $B_i$ satisfies property~\ref{prop:bachman_bijection}.

Assume now that $H_{i - 1}$ contains a hyperedge~$E'$ distinct from~$E$ with $E_{i - 1} \subseteq E'_{i - 1}$.
In that case, $\bigl[ \phi(E') \leadsto X \bigr]_{i - 1}$ and, thus,  $\bigl| \Phi(X) \bigr|_{i - 1} > 1$ (if $\phi(E')_{i - 1} = X$) or $X$ has incoming edge in~$B_{i - 1}$.
(As result, Algorithm~\ref{algo:bachman} calls line~\ref{line:bach_pendV_newSet}.)
Since $v$ is only added into~$E$ but not~$E'$, $E_i \nsubseteq E'_i$.
However, for all $\calE' \subseteq \calE_i$ with $E \in \calE'$ and $|\calE'| > 1$, $\Bigl[ \bigcap_{E \in \calE'} E \Bigr]_i = \Bigl[ \bigcap_{E \in \calE'} E \Bigr]_{i - 1}$.
Therefore, $\calX_i = \calX_{i - 1} \cup \{ \frkY \}$ with $\frkY_i = E_i$.
For each $\frkX \in \calX_{i - 1}$, let $f_i(\frkX) = f_{i - 1}(\frkX)$.
Additionally, let $f_i(\frkY) = Y$ where $Y = \{ v \}$ is the node added to~$B$ in line~\ref{line:bach_pendV_newSet}.
Thus, $f_i$ is a bijection mapping $\calX_i$ onto the nodes of~$B_i$.
Since the added edge~$(Y, X)$ points towards~$X$, $\bbV(Z)_i = \bbV(Z)_{i - 1}$ for all nodes~$Z$ of~$B_{i - 1}$ and $\bbV(Y)_i = \frkY$.
Hence, $B_i$ satisfies property~\ref{prop:bachman_bij_f}.
Additionally, since the algorithm also sets $\phi(E)_i = Y$, $B_i$ also satisfies property~\ref{prop:bachman_bij_phi}.

\paragraph{Case~\ref{case:xPendE}: \( x_i \) represents a hyperedge~\( E \in \calE \) and is pendant in~\( \calI_i \).}
Let $v$ be the vertex represented by the neighbour of~$x_i$ in~$\calI_i$ (\ie, $E_i = \{ v \}$), and let $X = \psi(v)_{i - 1}$.
Assume that $\calX_{i - 1}$ contains a set~$\frkX$ with $\frkX_{i - 1} = \{ v \}$.
In that case, adding~$E$ does neither change~$\calX$ nor any of the sets contained in it.
Additionally,  $|X_i| = 1$ and $X$ has no outgoing edges in~$B_{i - 1}$.
It follows that setting~$\phi(E)_i = X$ is the only change needed for~$B_i$ to satisfy property~\ref{prop:bachman_bijection} (similar to case~\ref{case:xTwinE}).
Algorithm~\ref{algo:bachman} does exactly that in line~\ref{line:bach_pendE_setPhi}.

Assume now that, for each set~$\frkX \in \calX_{i - 1}$, $\frkX_{i - 1} \neq \{ v \}$.
In that case, $\calX_i = \calX_{i - 1} \cup \{ \frkY \}$ with $\frkY_i = E_i = \{ v \}$.
Additionally,  $|X_i| > 1$ or $X$ has an outgoing edge in~$B_{i - 1}$.
Let $f_i(\frkX) = f_{i - 1}(\frkX)$ for each $\frkX \in \calX_{i - 1}$, and let $f_i(\frkY) = Y$ where $Y = \{ v \}$ is the node added to~$B$ in line~\ref{line:bach_pendE_newSet}.
Thus, $f_i$ is a bijection mapping $\calX_i$ onto the nodes of~$B_i$.
Note that Algorithm~\ref{algo:bachman} (in line~\ref{line:bach_pendE_newSet}) moves~$v$ from node~$X$ into the new node~$Y$.
However, since the added edge~$(X, Y)$ points towards~$Y$, $\bbV(Z)_i = \bbV(Z)_{i - 1}$ for all nodes~$Z$ of~$B_{i - 1}$ and $\bbV(Y)_i = \frkY$.
Therefore, due to the algorithm setting $\phi(E)_i = Y$, $B_i$ satisfies property~\ref{prop:bachman_bijection}.
\qed
\end{proof}

\begin{proof}[Complexity]
One can compute a pruning sequence for a given distance-hereditary graph in linear time~\cite{DamiHabiPaul2001} and, thus, for~$\calI(H)$ (line~\ref{line:bach_pruning}) in $\calO(N)$ time.
Creating $B$ and adding the first node (lines~\ref{line:bach_emptyB} and~\ref{line:bach_firstV}) can then be done in constant time.
For each node~$X$ of~$B$, we create two lists.
One stores the vertices in~$X$ and one the hyperedges in~$\Phi(X)$.
For the functions $\phi$ and~$\psi$, we store the node~$X$ they map on and a reference to where the hyperedge or vertex is stored in the corresponding list of~$X$.
That way, we can perform each of the following operations in constant time:
adding new nodes and edges into~$B$ (lines~\ref{line:bach_pendV_newSet} and~\ref{line:bach_pendE_newSet}),
assigning a hyperedge to a node and setting~$\phi$ (lines~\ref{line:bach_twinE_setPhi}, \ref{line:bach_pendE_setPhi}, and~\ref{line:bach_pendE_newSet}),
changing the assignment of a hyperedge to a different node and updating~$\phi$ (line~\ref{line:bach_pendV_newSet}),
adding a vertex into a node and setting~$\psi$ (lines \ref{line:bach_twinV_addV}, \ref{line:bach_pendV_addV}, and~\ref{line:bach_pendV_newSet}), and
moving a vertex from one node into another and updating~$\psi$ (line~\ref{line:bach_pendE_newSet}).
Therefore, each iteration of the loop starting in line~\ref{line:bach_loop} run in constant time and, subsequently, Algorithm~\ref{algo:bachman} run in overall linear time.
\qed
\end{proof}

\begin{lemma}
\label{lem:BachmanDegree}
Each node~\( X \) of~\( B \) with \( \Phi(X) = \emptyset \) has an in-degree of at least~\( 2 \).
\end{lemma}

\begin{proof}
We first assume that $X$ has in-degree~$0$.
Then, there is no hyperedge~$E$ with $\phi(E) \leadsto X$ and, subsequently, no such~$E$ with $\bbV(X) \subseteq E$.
That contradicts with the definition of Bachman diagrams.

Now assume that $X$ has at least one incoming edge $(Y, X)$.
Let $\calE_X = \bigl \{ \, E \bigm| \phi(E) \leadsto X \, \bigr \}$ and $\calE_Y = \bigl \{ \, E \bigm| \phi(E) \leadsto Y \, \bigr \}$.
Since $\bigcap_{E \in \calE_X} E = \bbV(X) \subset \bbV(Y) = \bigcap_{E \in \calE_Y}$, there is a hyperedge~$E \in \calE_X \setminus \calE_Y$ with $\phi(E) \leadsto X$ and $\phi(E) \not\leadsto Y$.
Hence, since $\phi(E) \neq X$, there is a path from $\phi(E)$ to~$X$ in~$B$ that does not contain~$Y$ and, therefore, $X$ has an in-degree of at least~$2$.
\qed
\end{proof}

\begin{algorithm}[!htb]
\caption
{%
    Computes a subset graph for a given $\gamma$-acyclic hypergraph.
}
\label{algo:gammaSubset}

\KwIn
{
    A $\gamma$-acyclic hypergraph~$H = (V, \calE)$.
}

\KwOut
{
    A subset graph~$G$ for~$H$.
}

Compute a simplified Bachman diagram~$B$ for $H$ with the corresponding functions $\phi$ and~$\Phi$ (see Algorithm~\ref{algo:bachman}).
\label{line:gammaSubset_bachman}

Create a new directed graph~$G = (\calE, \emptyset)$.
\label{line:gammaSubset_createG}

\ForEach
{%
    \( E \in \calE \)
}
{%
    Let $X = \phi(E)$.
    Compute $\bbE_X =  \bigcup_{Y \leadsto X} \Phi(Y)$.
    \label{line:gammaSubset_findEdges}

    For each $E' \in \bbE_X$ distinct from~$E$, add the edge~$(E, E')$ to~$G$.
    \label{line:gammaSubset_addEdge}
}
\end{algorithm}

\begin{theorem}
Algorithm~\ref{algo:gammaSubset} computes the subset graph~\( G \) of a given \( \gamma \)-acyclic hypergraph in \( \calO \bigl( N + |G| \bigr) \) time.
\end{theorem}

\begin{proof}[Correctness]
Let $E$ and~$E'$ be two distinct hyperedges of~$H$.
By definition of (simplified) Bachman diagrams, $B$ (computed in line~\ref{line:gammaSubset_bachman}) contains two nodes $X = \phi(E)$ and~$Y = \phi(E')$ such that $Y \leadsto X$ if and only if $E \subseteq E'$.
Additionally, Algorithm~\ref{algo:gammaSubset} adds the edge~$(E, E')$ to~$G$ (line~\ref{line:gammaSubset_addEdge}) if and only if $Y \leadsto X$.
Therefore, for any distinct hyperedges $E$ and~$E'$ of~$H$, $(E, E')$ is an edge of~$G$ if and only if $E \subseteq E'$.
\qed
\end{proof}

\begin{proof}[Complexity]
Computing the simplified Bachman diagram~$B$ (line~\ref{line:gammaSubset_bachman}) can be done in $\calO(N)$ time (Lemma~\ref{lem:bachmanAlgo}).
Creating the graph~$G$ (line~\ref{line:gammaSubset_createG}) can be done in $\calO(m)$ time.
Additionally, once the sets~$\bbE_X$ are known for all~$X$, we can add the edges of~$G$ (line~\ref{line:gammaSubset_addEdge}) in $\calO \bigl( |G| \bigr)$ total time.
It remains to show that we can compute the sets~$\bbE_X$ in the desired runtime.
To do that, we show that we can compute $\bbE_X$ for a given~$X$ in $\calO \bigl( |\bbE_X| \bigr)$ time.

Recall that $B$ is a directed graph which forms a tree (Lemma~\ref{lem:bachmanTree}).
Hence, the the nodes of~$B$ from which there is a path to~$X$ form a tree~$T_X$ rooted in~$X$ where each edge points from a child to its parent.
One can compute $T_X$ in $\calO \bigl( |T_X| \bigr)$ time by, for example, reversing the edges of~$B$ and then performing a BFS or DFS starting at~$X$.

Assume that we partition the nodes of~$T_X$ into two sets $\bbY_{\!X}$ and~$\bbZ_X$ where $\bbY_{\!X} = \bigl \{ \, Y \bigm| \Phi(Y) \neq \emptyset \, \}$ and~$\bbZ_X$ contains all remaining nodes.
It follows from Lemma~\ref{lem:BachmanDegree} that each node~$Y$ of~$T_X$ with at most one child (including leaves) is in~$\bbY_{\!X}$, and each node in~$\bbZ_X$ has at least two children.
Now assume that we, step by step, remove each node~$Y$ from~$T_X$ which has exactly one child~$Y'$, and make $Y'$ the child of $Y$'s parent.
Let $T'_X$ be the resulting tree.
Each node of~$T'_X$ then has at least two children.
Thus, at least half of the nodes of~$T'_X$ are leaves.
Since each leaf is in~$\bbY_{\!X}$ and $T'_X$ contains all nodes in~$\bbZ_X$, it follows that $|\bbZ_X| \leq |\bbY_{\!X}|$ and, subsequently, $|T_X| \in \Theta \bigl( |\bbY_{\!X}| \bigr)$.

Recall that $\Phi(Y) \neq \emptyset$ for all $Y \in \bbY_{\!X}$ and that each hyperedge of~$H$ is associated with at most one such~$Y$.
It follows that $|\bbY_{\!X}| \leq |\bbE_X|$.
Therefore, we can compute $\bbE_X$ for a given~$X$ in $\calO \bigl( |\bbE_X| \bigr)$ time, and line~\ref{line:gammaSubset_findEdges} runs in $\calO \bigl( |G| \bigr)$ total time.
\qed
\end{proof}

\section{Interval Hypergraphs}
\label{sec:interval}

An acyclic hypergraph~$H = (V, \calE)$ is an \emph{interval hypergraph} if it admits a join tree that forms a path.
That is, there is an order~$\sigma = \langle E_1, E_2, \ldots, E_m \rangle$ for the hyperedges of~$H$ such that, for each vertex~$v \in V$, $v \in E_i \cap E_j$ implies that $v \in E_k$ for all~$k$ with $i \leq k \leq j$.
Interval hypergraphs are closely related to interval graphs which are a subset of chordal graphs.
In particular, a graph is an interval graph if and only if its maximal cliques form an interval hypergraph, and an acyclic hypergraph is an interval hypergraph if and only if its 2-section graph is an interval graph.

Algorithm~9 in~\cite{HabMcCPauVie2000} allows to recognise interval hypergraphs in linear time.
It also produces an order~$\sigma$ as defined above.
Note that the first step of that algorithm is to compute a clique tree~$T$ and a vertex ordering~$\phi$ for a given graph.
We replace that step by first computing a join tree~$T$ of the given hypergraph and then perform Algorithm~10 from~\cite{HabMcCPauVie2000} to compute~$\phi$.

There are multiple ways to compute the subset graph and union join graph once $\sigma$ is known for a given hypergraph~$H$.
One may order the vertices of~$H$ based on the right-most hyperedge containing them (with respect to~$\sigma$).
Note that we can compute such an ordering in linear time from~$\sigma$.
Both orders together then form a doubly lexically order and allow to construct a $\Gamma$-free matrix for~$H$.
Note that Algorithm~\ref{algo:betaSubset} and the algorithm described in Theorem~\ref{theo:betaUnionJoin} only have a logarithmic overhead in runtime because they compute a doubly lexically order.
If such an order is given, both algorithm run in $\calO \bigl( N + |G| \bigr)$ time.

For an alternative approach, we first determine for each vertex~$v$ the index of the left-most hyperedge containing it (with respect to~$\sigma$).
Let $\phi(v)$ be that number, \ie, if $E_i$ is the left-most hyperedge containing~$v$, then $\phi(v) = i$.
Next, we compute the separators between consecutive hyperedges (see Algorithm~\ref{algo:modifyT}).
Let $S_i$ denote the separator between $E_{i - 1}$ and~$E_i$ and let $\phi(S_i) = \max_{v \in S_i} \phi(v)$.
Then, for each $E_j$ with $j < i$, it holds that (i)~$E_j \supseteq E_i$ if and only if $|E_i| = |S_i|$ and $j \geq \phi(S_i)$, and (ii)~$E_iE_j$ is an edge of the union join graph of~$H$ if and only if $j \geq \phi(S_i)$.
Running the same approach again using the reverse of~$\sigma$ therefore allows to compute the subset graph and union join graph in $\calO \bigl( N + |G| \bigr)$ time.

\begin{theorem}
There is an algorithm that computes the union join graph and subset graph of a given interval hypergraph in \( \calO \bigl( N + |G| \bigr) \) time, respectively, where \( |G| \) is the size of the computed graph.
\end{theorem}

\section*{Acknowledgements}
We would like to thank Feodor F.\ Dragan and Rachel Walker for stimulating discussions.

\end{document}